\documentclass[a4paper,12pt]{article}
\usepackage{JASA_manu}
\usepackage{amsmath}
\usepackage{latexsym}
\usepackage{amssymb}
\usepackage{amsfonts}
\usepackage{float}
\usepackage{bm}
\usepackage{color}
\usepackage{srcltx}
\usepackage{multirow}
\usepackage{graphicx}
\usepackage{natbib}
\usepackage{amsthm}
\usepackage{wasysym}
\usepackage{url}
\newtheorem{proposition}{Proposition}

\def\E{\mathbb{E}}
\def\var{\mathbb{V}ar}

\def\cor{\mathbb{C}or}

\def\R{{\rm I\! R}}
\def\N{{\rm I\! N}}
\def\E{{\rm I\! E}}

\title{On modeling positive continuous data
	with spatio-temporal dependence}

\author{Moreno Bevilacqua \\
			 Department of Statistics\\
		Universidad de Valparaiso,
 Valparaiso, Chile\\
 and\\
		Millennium Nucleus Center \\ for the Discovery of Structures in  Complex Data,
		Chile
			   \\
	Christian Caama\~no-Carrillo\\
	Department of Statistics, \\
		Universidad del B\'io-B\'io,   Concepcion, Chile
	\\
	Carlo Gaetan\\
	Dipartimento di Scienze Ambientali, Informatica e Statistica, \\
		Universit\`a  Ca' Foscari di Venezia, 		 Venice, Italy.
	   }

\begin{document}
	\maketitle
	\newpage
\begin{abstract}
In this paper, we concentrate  on an alternative modeling strategy for positive data that exhibit  spatial or spatio-temporal dependence.
Specifically, we propose to consider  stochastic processes obtained through a  monotone transformation of  scaled version of $\chi^2$ random processes. The latter are well known in the specialized literature and originates by   summing independent copies of a squared Gaussian  process. However, their use as stochastic models and related inference has not been much considered.
Motivated by a spatio-temporal analysis of wind speed data
from a network of meteorological stations in the Netherlands, we exemplify our modeling strategy by means of a non-stationary process with Weibull marginal distributions.
For the proposed Weibull process we study the second-order and geometrical properties  and  we provide  analytic expressions for the bivariate distribution.
Since the  likelihood is intractable,  even for a relatively small data-set, we suggest adopting  the  pairwise likelihood as a tool for inference. 
Moreover, we tackle the prediction problem and we propose {to use} a  linear prediction. The effectiveness of our modeling strategy  is illustrated by analyzing  the aforementioned Netherland wind speed data 
that we integrate with a simulation study. The proposed method is implemented in the R package {\texttt{GeoModels}}.
\end{abstract}
\newpage
\section{Introduction}

Climatology, Environmental Sciences  and Engineering, to name some fields, show  an increasing interest in the statistical  analysis of
spatial and/or   temporal  data.
In order to model the inherent uncertainty of the data,
 Gaussian random processes play a fundamental role \citep[see][for instance]{Cressie:Wikle:2011}.
Indeed,  the Gaussian random processes have to offer marginal and dependence modelling in terms of  mean and covariance functions,  methods of  inference well studied and scalable for
large dataset \citep{Heaton2018} and  optimality in the prediction \citep{Stein:1999}.

However, data collected in a range of studies
such as wind speeds  \citep{Pryor:Barthelmie:2010}, ocean surface
currents  \citep{Galanis_et_al:2012} and rainfalls
\citep{Neykov:Neytchev:Zucchini:2014}
take continuous positive values and exhibit  skewed
sampling distributions.
In this case the
Gaussian probability model becomes  unrealistic.

Transformations of a Gaussian
 process, i.e. trans-Gaussian kriging \citep{Cressie:1993},   is a general  approach  to model
this kind of data  by applying a nonlinear
transformations to the original data. In the literature the most common transformations,
the square root  and the natural logarithm, belong to the Box-Cox power transformation 
\citep[see][for instance]{Haslett:Raftery:1989,Allcroft:Glasbey:2003,Bessac:Ailliot:Monbet:2015}.
In particular, Log-Gaussian processes have been  broadly used  for the analysis  of positive dependent  data 
due to their  well known mathematical properties \citep{Oliveira_et_al:1997,DeOliveira:2006}.
For an alternative (parametric) family of transformations of Gaussian processes, the Tukey $g$-and-$h$ transformation, see \citet{Xu:Genton:2017}, \citet{Yan-Genton:2019}.

Nevertheless, it can be difficult to find an adequate non linear
transformation and some appealing properties of the
Gaussian process  may not be inherited by the transformed process \citep{Wallin:Bolin:2015}.

Another possibility is to resort to Gaussian copulas. Copula theory \citep{Joe:2014}  allows joint distributions to be constructed from specified marginal continuous distributions for positive data. 
The role of the copula is to describe the  spatio-temporal dependence structure between random variables without information on the marginal distributions.
Even though which copula model to use for a given analysis is not generally known a priori, the copula based on the multivariate Gaussian distribution
\citep{Kazianka:Pilz:2010,Masarotto:Varin:2012,BG:2014} has gained a general consensus  since the definition of the multivariate  dependence  relies again on the specification of the pairwise dependence, i.e. on the   covariance function.

Actually the two aforementioned  approach are strongly related since    monotone transformations of a Gaussian process share the same copula model. As we will see in our real data example, the kind of   dependence described by the Gaussian copula could be   {restrictive or unsuitable}.
In fact the  Gaussian copula 
expresses a symmetrical dependence  (see Section \ref{sec:models}), 
i.e. high values  exhibit a  spatial/temporal dependence similar to  low ones.
Copula-based model using symmetrical dependence is still used in a  recent paper \citep{Tang:Wang:Sun:Hering:2019} on spatio-temporal modelling wind speed data.

Concluding  this short review we mention  that
\cite{Wallin:Bolin:2015} proposed recently non-Gaussian processes derived
from stochastic partial differential equations. 
Nevertheless, their method is
restricted to the spatial Mat\'ern covariance model and its statistical
properties are much less understood that the Gaussian process.

In this paper, we shall look at  processes that are derived by  Gaussian processes but differently from the  trans-Gaussian random processes and the copula models we do not consider just one copy of the Gaussian process.
We suggest to model positive continuous data by  transforming $\chi^2$ processes  \citep{Adler:1981,MA2010}
 i.e.   a sum of squared of independent copies of a  standard Gaussian  process. Even though probabilistic properties of a sum of squared Gaussian processes have been  studied
 	several years ago, less attention has been paid to use this for statistical modelling   of dependent positive data.
We are convinced that the Gaussian processes offer an incomparable tool case for those who want to model the dependence between observations. However, we aim to overcome some aforementioned restrictions.

Motivated by a spatio-temporal analysis of  daily wind speed data 
from a network of meteorological stations in the Netherlands,
we  {exemplify} our construction
by proposing a non-stationary spatio-temporal process with asymmetric dependence 
and    Weibull marginal   distribution even though  other stochastic processes with different marginal distributions could be studied starting  from transformations of  $\chi^2$ processes. 
In fact,  in scientific literature a variety of probability distribution has been suggested to describe wind speed  distributions and the Weibull model constitutes one of the most widely accepted  \citep[see][for a review]{Carta:Ramirez:Velazquez:2009}.

The proposed Weibull  process is parametrized in such a way 
that both  regression and dependence analysis can be jointly performed.   
Additionally the  process  inherits the geometrical properties of the   underlying Gaussian process.
This implies that  mean-square continuity and differentiability,  as in the Gaussian processes,    can be modeled using  suitable parametric
correlation models such as the  Mat{\'e}rn \citep{Stein:1999} or the Generalized Wendland \citep{Bevilacqua:20189} models,   in the spatial case.

It must be said that it is the difficult to evaluate  the multivariate density for the  proposed model and this fact  prevents the inference  based on the  full  likelihood and the derivation of the analytical form of the predictor that minimizes  the mean square prediction error. 
For this reason,  we propose the use of a weighted version of the pairwise likelihood
\citep{Lindsay:1988,Varin:Reid:Firth:2011}
for estimating the unknown parameters.
Moreover a   linear and unbiased predictor  is proposed following the approach detailed in \citet{DeOliveira:2014}.

 {In addition, we study a specific example where both the multivariate density and the optimal predictor have  an explicit closed-form
(see Section 5).
For this  example, we perform a simulation study with the goal of investigating  the relative efficiency
of the  maximum weighted  pairwise  likelihood method with respect to the maximum likelihood method
and the relative efficiency  of the proposed linear  predictor with respect to the  optimal predictor.
}
Simulation, estimation, and prediction of the Weibull process  are implemented in a R package
	\texttt{GeoModels} \citep{Bevilacqua:2018aa}.

The remainder of the paper is organized as follows. In Section \ref{sec:models}
we introduce the random processes and we describe their features.
In Section \ref{sec:Weibull} we concentrate on a random process with Weibull marginal distributions.
Section \ref{sec:est-pred} starts with a short description of  the estimation method  and ends with tackling  the prediction problem.
In Section \ref{sec:numerical-results} we report the numerical results of a  simulation study and in Section \ref{sec:real-data} we apply our method to the  daily wind speed data measurements 
from a network of meteorological stations in the Netherlands using the Log-Gaussian process as the benchmark.
Finally some concluding remarks are consigned to Section \ref{sec:conclusions}.

\section{Scaled $\chi^2$ random processes}\label{sec:models}

\subsection{Definition}
We start by considering a `parent'  Gaussian random process $Z:=\{Z(s), s \in S  \}$, where $s$ represents  a location in the domain $S$.  Spatial ($S\subseteq \R^d$) or spatio-temporal examples ($S\subseteq \R^d\times\R_+$) will be considered indifferently.
We also assume that $Z$ is stationary with  zero mean, unit variance and correlation function $\rho(h):=\cor(Z(s+h),Z(s))$ where $s+h\in S$.
Let $Z_1,\ldots,Z_m$ be
$m=1,2,\ldots$ independent copies of $Z$ and define the  random {process} $X_m:=\{X_m(s), s \in S \}$ as
\begin{equation}\label{gg}
X_{{m}}({s}):=\sum_{k=1}^{{m}} Z_k({s})^2/{m}.
\end{equation}
The stationary process $X_{{m}}$ is a  scaled version of a $\chi^2$  random process \citep{Adler:1981,MA2010}   with
marginal distribution  $\mathrm{Gamma}(m/2,m/2)$ where
the  pairs ${m}/2,m/2$ are the  shape  and rate parameters. By definition, $\E(X_{{m}}({s}))=1$ and  $\var(X_{{m}}({s}))=2/{m}$ for all $s$.

The analytical expressions of  the multivariate density of a vector of $n$ observations $X_m(s_1)=x_1,\ldots,X_m(s_n)=x_n$   can be derived only in some special cases \citep{krishnamoorthy1951,Royen:2004}.
An interesting example is made up for  $s_1<s_2<\ldots< s_n$ locations on $S=\mathbb{R}$ and for a Gaussian process $Z$ with exponential covariance function. In this case the multivariate density can be derived as 
\begin{eqnarray}\label{gammafd1}
f_{X_m}(x_1,\ldots,x_n)&=&\frac{m^{{m}/{2}-1+n}2^{-{m}/{2}+1-n}(x_1x_n)^{{m}/{4}-{1}/{2}}}{\Gamma\left({m}/{2}\right)\prod\limits^{n-1}_{i=1}\{(1-\rho^2_{i,i+1})\rho_{i,i+1}^{{m}/{2}-1}\}}
\nonumber\\
&&\times
\exp{\left[-\frac{mx_1}{2(1-\rho^2_{1,2})}
-\frac{mx_n}{2(1-\rho^2_{n-1,n})}-\sum\limits_{i=2}^{n-1}\frac{m(1-\rho^2_{i-1,i}\rho^2_{i,i+1})x_i}{2(1-\rho^2_{i-1,i})(1-\rho^2_{i,i+1})}\right]}\nonumber\\
&&\times\prod\limits^{n-1}_{i=1}I_{{m}/{2}-1}\left(\frac{m\rho_{i,i+1}\sqrt{x_ix_{i+1}}}{(1-\rho^2_{i,i+1})}\right)
\end{eqnarray}
with
$\rho_{ij}:=\exp\{-|s_i-s_j|/\phi\}$, $\phi>0$ and $I_{a}(x)$  the modified Bessel function of the
first kind of order $a$.

 {The evaluation} of the bivariate densities  of a pair  of observations $X_m(s_1)$ and $X_m(s_2)$  
can be derived
irrespective of the dimension of the space  $S$ and the
correlation function \citep{VJ:1967}. 
The bivariate distribution of $X_m$   is  known as the  Kibble bivariate Gamma distribution
\citep{Kibble:1941}
with density
\begin{eqnarray}\label{eq:kibble}
f_{X_m}(x_1,x_2)&=&\frac{2^{-{m}}{m}^{{m}}(x_1x_2)^{{m}/2-1}
}{\Gamma\left({{m}}/{2}\right)(1-\rho^2)^{{m}/2}}
\left(\frac{{m}|\rho|\sqrt{x_1x_2}}{2(1-\rho^2)}\right)^{1-{m}/2}
\exp\left\{-\frac{{m}(x_1+x_2)}{2(1-\rho^2)}\right\}\nonumber
\\
&&\times\, I_{{m}/2-1}
\left( \frac{ {m}|\rho|\sqrt{x_1x_2} } {(1-\rho^2)}\right),
\end{eqnarray}
where $\rho=\rho(s_1-s_2)$.

\subsection{Dependence structure}
It is easy to show that
the correlation function of $X_m$, $\rho_{X_m}(h)$, is equal to  $\rho^2(h)$, the squared of the correlation function of the `parent' Gaussian random process.
However, a way for  looking more deeply into the dependence structure between random variables   {regardless of}  the marginal distributions is inspecting their copulas \citep{Joe:2014}.
For a $n$-variate cumulative distribution function (cdf) $F(y_1,\ldots,y_n):=\Pr(Y_1\le y_1,\ldots,Y_n\le y_n)$
with $i$-th univariate margin $F_i(y_i):=\Pr(Y_i\le y_i)$, the copula associated with $F$ is a  cdf function $C : [0, 1]^n\rightarrow [0, 1]$ with
$\mathcal{U}(0,1)$ margins that satisfies $F(y_1,\ldots,y_n)=C(F_1(y_1),\ldots,F_n(y_n)).$
If $F$ is a  {continuous cdf, with quantile functions $F_i^{-1}$, $i=1,\ldots,n$,}   then Sklar's theorem \citep{Sklar:1959} guarantees
that the cdf   {on the $n$-hypercube} %
$
C(u_1,\ldots,u_n)=F(F_1^{-1}(u_1),\ldots, F_n^{-1}(u_n))
$
is the unique choice. The corresponding copula density, obtained by differentiation, is denoted by $c(u_1,\ldots,u_n)$.

Analogously, the multivariate survival function $\overline{F}(y_1,\ldots,y_n):=\Pr(Y_1> y_1,\ldots,Y_n> y_n)$
 could be expressed using the univariate survival functions $\overline{F}_i=1-F_i$ and the survival copula
$\overline{F}(y_1,\ldots,y_n)=\overline{C}\,(\overline{F}_1(y_1),\ldots,\overline{F}_n(y_n)).$
 {The survival copula is in this way a distribution function on the $n$-hypercube
$
\overline{C}(u_1,\ldots,u_n)=\overline{F}(\overline{F}^{-1}(u_1),\ldots,\overline{F}^{-1}(u_n)).
$}

Among the copulas, the Gaussian copula is a convenient model for spatial data \citep{Masarotto:Varin:2012}  as it  offers a parametrization in terms of a correlation function.
Let $\Phi^{-1}$ denote the quantile   function of $\Phi$  the cdf of a standard Gaussian variable. The Gaussian copula with correlation matrix $R$  is defined by
$
C(u_1,\ldots,u_n)=\Phi_R(\Phi^{-1}(u_1),\ldots, \Phi^{-1}(u_n)),
$
where $\Phi_R$ denotes the joint cumulative distribution function of an $n$-variate Gaussian random vector with zero means and correlation matrix $R$.

 The Gaussian copula  is reflection symmetric, 
$
C(u_1,\ldots,u_n)=\overline{C}(u_1,\ldots,u_n)
$, that is the  probability of having all variables less than their respective $u$-th quantile is  {equal to} the probability of having all the variables greater than the complementary marginal quantile.
Such property is a potential issue for data in which  upper quantiles    {might} exhibit a different pairwise spatial dependence than lower quantiles.

 {Although copula theory uses transformations to $\mathcal{U}(0, 1)$ margins, it is better to consider $\mathcal{N}(0, 1)$ margins \citep[p. 9]{Joe:2014} for identifying the copula and for diagnostic purpose.}
In particular, the plot of the bivariate copula density can be compared with the Gaussian bivariate density and with the scatter-plot of pairs of observations on a normal scale, the \textit{normal scores}.
The bivariate copula density with $\mathcal{N}(0, 1)$ margins is given by
$$
c_{N}(z_1,z_2)=\frac{c(\Phi^{-1}(z_1),\Phi^{-1}(z_2))}{\phi(z_2)\phi(z_2)}
$$
where $\Phi(z)$ (resp. $\phi(z)$)  is the cdf  (resp. pdf) of the standardized Gaussian distribution.
Under this transform  reflection symmetry  means that the
bivariate density contour plot is symmetric to the $(z_1 , z_2 )\rightarrow (-z_1 , -z_2 )$ reflection, i.e.
$
c_{N}(z_1,z_2)=c_{N}(-z_1,-z_2)
$.

In Figure \ref{fig:copula} we  compare the contour plots of the  bivariate copula density function
entailed by (\ref{eq:kibble})
 with  the elliptical contours of the bivariate Gaussian density.
  Note that the copula for $m=1$ is the copula introduced in \citet{Bardossy:2006}.
 Sharper corners (relative to ellipse) indicate more tail dependence of $X_m$ than the Gaussian process and we notice  also reflection asymmetries.

 {Dependence among extremal events can be summarized by the  \emph{upper tail dependence coefficient} \citep{Sibuya:1960, Coles:Heffernan:Tawn:1999} that looks at the limit behavior of the
conditional probability  of two random variables $Y_1$ and $Y_2$ 
 $$
 \tau:= \lim_{u \rightarrow 1^-}
 \frac{\Pr(F_1(Y_1)> u,F_2(Y_1)> u)}{\Pr(F_2(Y_1)> u)} 
= \lim_{u \rightarrow 1^-}
 \frac{1-2u+C(u,u)}{1-u}= 
 \lim_{u \rightarrow 1^-}
 \frac{\overline{C}(1-u,1-u)}{1-u}.
 $$
}
 The  value of the coefficient helps to distinguish  between asymptotic dependence and asymptotic independence of the observations  as the quantile increases. Under spatial asymptotic dependence, the likelihood of a large event happening in one location is tightly related to high values being recorded at a location nearby; the opposite is true under asymptotic independence, in which large events might be recorded at one location only and not in neighboring locations \citep{Wadsworth:Tawn:2012}.

We say that $Y_1$ and $Y_2$ are asymptotically
dependent if $\tau >0$ is positive.
The case  $\tau=0$ characterizes asymptotic independence.
Simply adapting Theorem 2.1 in \citet{hashorva2014} we can prove that the  $X_m$ is asymptotically
independent for all $m$, i.e. $\tau=0$.

\begin{figure}
\begin{center}
\hspace{-1cm}
\scalebox{1}{
\begin{tabular}{cc}
\hspace{-.3cm}
\includegraphics[width=0.45\linewidth, height=0.27\textheight]{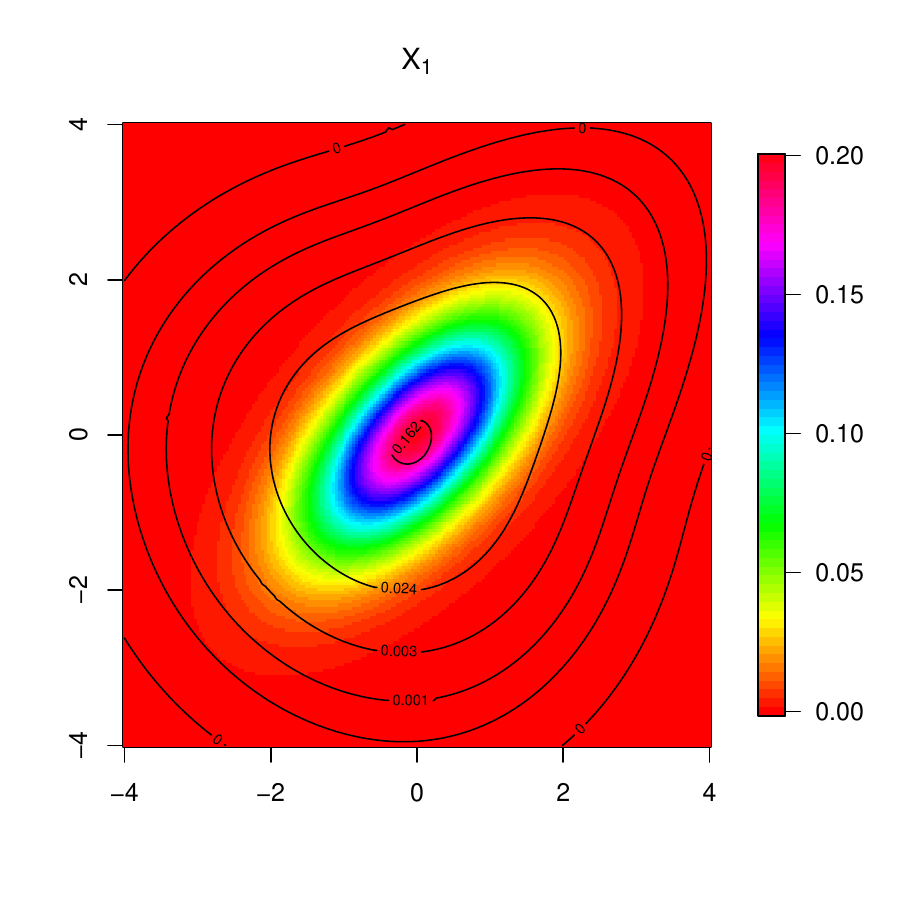}  &
\hspace{-.2cm}
\includegraphics[width=0.45\linewidth, height=0.27\textheight]{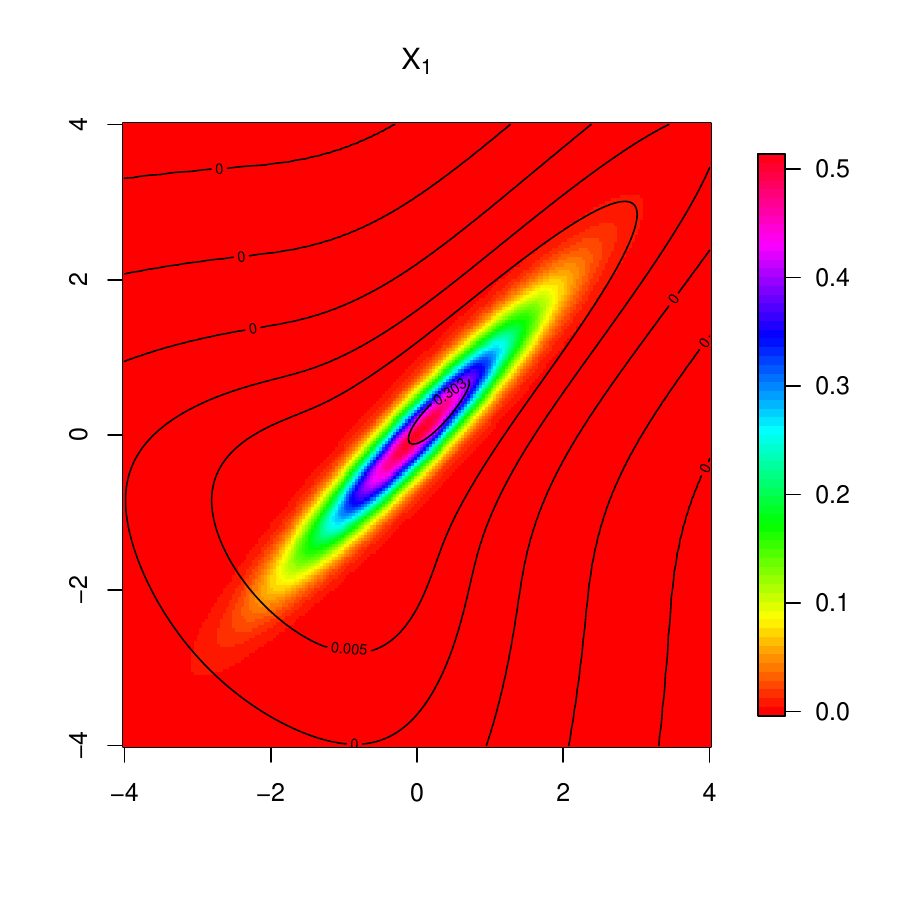}
\\
\hspace{-.3cm}
\includegraphics[width=0.45\linewidth, height=0.27\textheight]{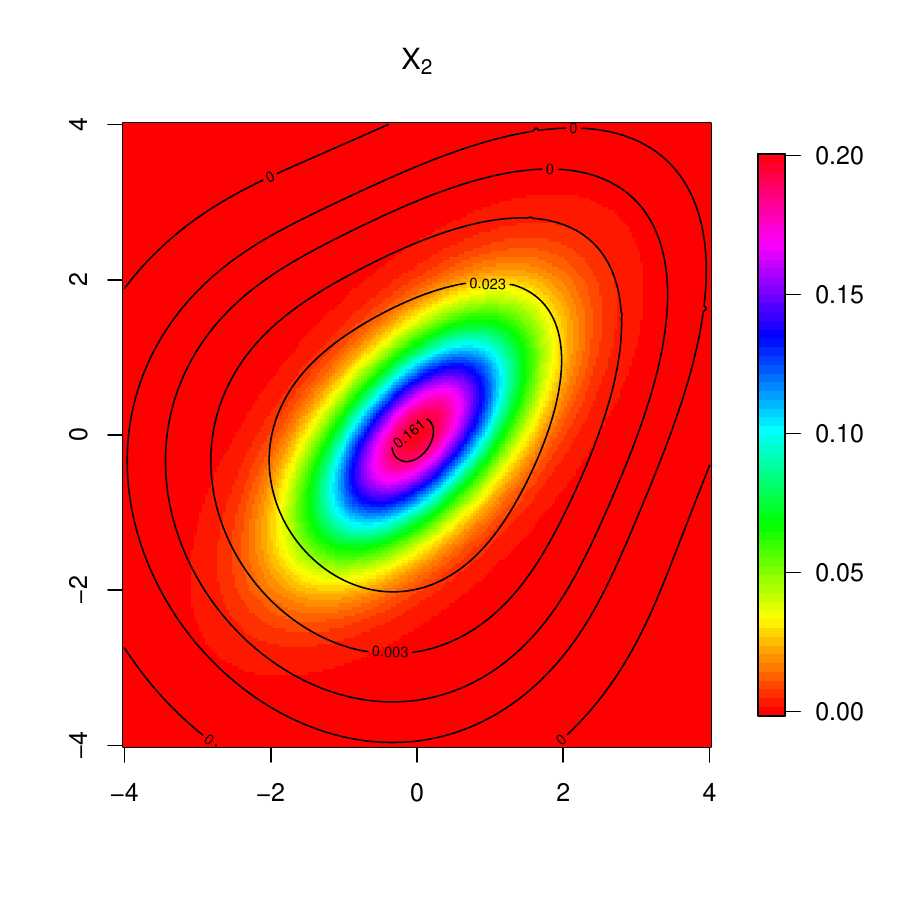}  &
\hspace{-.2cm}
\includegraphics[width=0.45\linewidth, height=0.27\textheight]{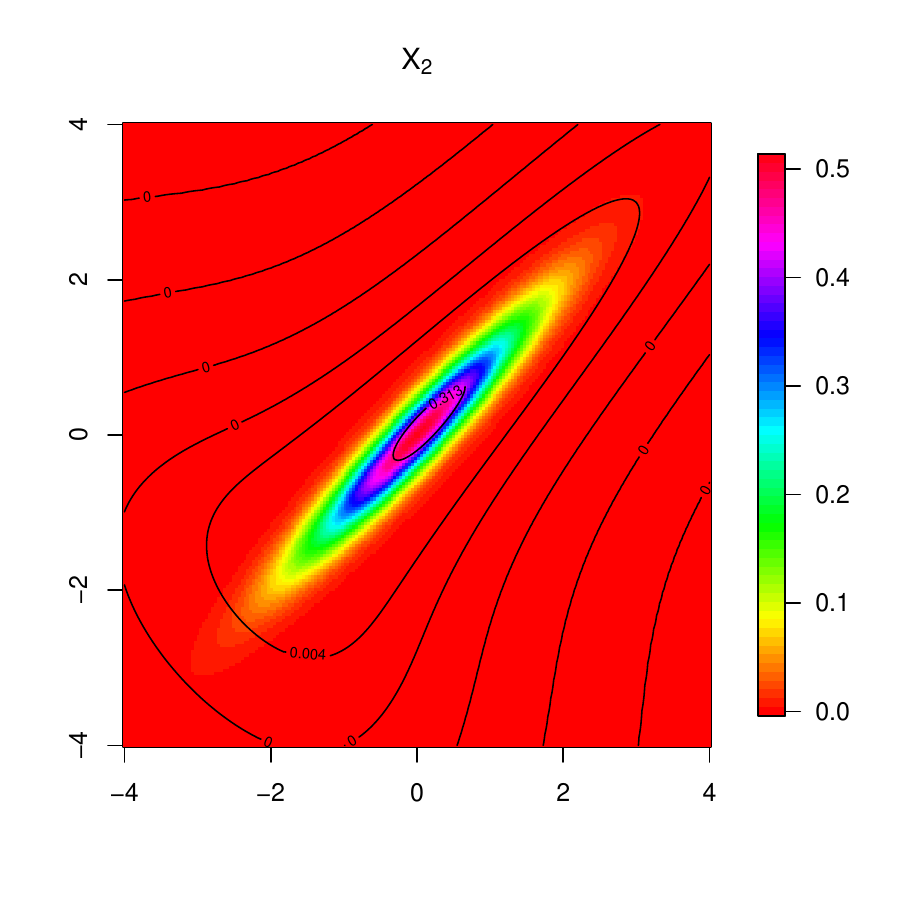}
\\
\hspace{-.3cm}
\includegraphics[width=0.45\linewidth, height=0.27\textheight]{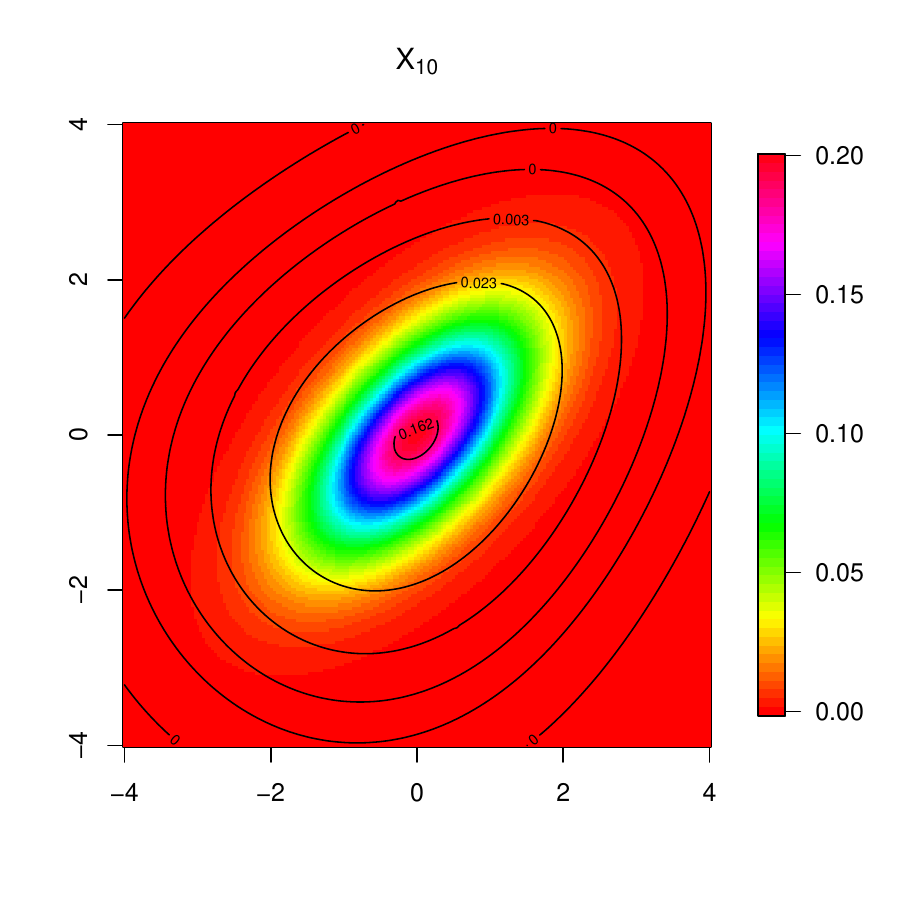}  &
\hspace{-.2cm}
\includegraphics[width=0.45\linewidth, height=0.27\textheight]{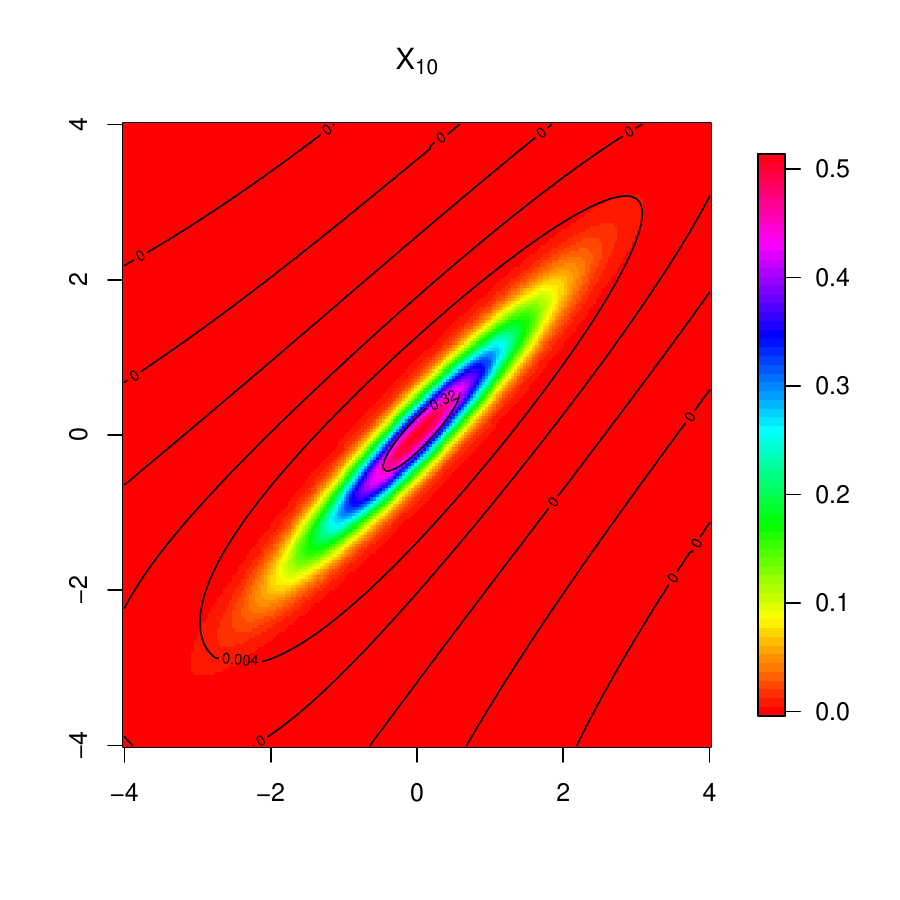}
\end{tabular}
 }
\end{center}
\caption{$X_m$ process: bivariate density contour plots for different values of $m$ and $\rho$ after transforms to $\mathcal{N}(0, 1)$ margins with $m=1,2,10$ from  top to  bottom and  $\rho=0.6,0.95$ from  left to  right.
The background image is a grid of colored pixels with colors corresponding to the values of the  standard bivariate Gaussian density with correlation $\rho$.}
\label{fig:copula}
\end{figure}

\section{A random process with Weibull marginal distribution}\label{sec:Weibull}

We focus our attention on stochastic modeling  of wind speed data. In scientific literature, a variety of probability distribution has been suggested to describe wind speed frequency distributions
\citep[see][for a review]{Carta:Ramirez:Velazquez:2009}. Among them, the Weibull distribution constitutes one of the most widely accepted distribution for wind speed and it can be derived from a physical argument.

Suppose that the  two orthogonal wind  components $(Z_1,Z_2)$
are assumed to be individually Gaussian with zero mean and independent, isotropic fluctuations.
The  distribution of   the speed  $V= \sqrt{Z_1^2+Z_2^2}$ is the
Rayleigh distribution i.e. the distribution $V^2$  is the exponential distribution \citep[pag. 417]{Johnson:Kotz:Balakrishnan:1995}.
 {The Weibull distribution can be obtained  from the  Rayleigh distribution through the power law transformation of $V$ %
that  has been shown to fit better wind speed samples due to its  flexible form induced by the additional shape parameter.}

 {Following this idea, let  $X_{{2}}$ a special case of the  rescaled $\chi^2$  random process defined in Equation (\ref{gg}), with  standard exponential   marginal distribution.
To obtain a stationary positive random {process} $W=\{W(s),s\in S\}$ with marginal Weibull distribution
we consider   the  power transformation
\begin{equation}\label{def:weibull}
{W}({s}):=
\nu(\kappa) X_{2}({s})^{1/\kappa},
\end{equation}
where $\nu(\kappa)= \Gamma^{-1}(1+1/\kappa)$  and $\kappa >0$ is a shape parameter. Note that under  this specific  parametrization
${W}({s})\sim\mbox{Weibull}( \kappa, \nu(\kappa))$,
$\E(W(s))=1$ and 
$\var(W(s))=(\Gamma\left(1+2/ \kappa\right)\nu^2(\kappa)-1)$.}
In addition, the density of a pair of observations $W(s_1)=w_1$ and $W(s_2)=w_2$
is easily obtained from
(\ref{eq:kibble}) and (\ref{def:weibull}), namely
\begin{eqnarray}\label{eq:weibull}
f_{W}(w_1,w_2)&=&\frac{\kappa^2(w_{1}w_{2})^{\kappa-1}}{\nu^{2\kappa}(\kappa)(1-\rho^2)}
\exp\left[-\frac{w_1^{\kappa}+w_2^{\kappa}}{\nu^\kappa (\kappa)(1-\rho^2)}\right]
I_{0}\left(\frac{2|\rho|(w_1w_2)^{\kappa/2}}{\nu^\kappa (\kappa)(1-\rho^2)}\right).
\end{eqnarray}

Using Proposition \ref{def:prop1} in Appendix
  we can also obtain the correlation function of $W$, namely

 \begin{equation}\label{ccc}
 \rho_{W}(h)=\frac{ \nu^{-2}(\kappa)    }{\left[\Gamma\left(1+{2}/{\kappa}\right)-\nu^{-2}(\kappa) \right]}\left[{}_2F_1\left(-{1/\kappa},-{1/\kappa};1;\rho^2(h)\right)-1\right],
 \end{equation}
where the function
$${}_pF_q(a_1,a_2,\ldots,a_p;b_1,b_2,\ldots,b_q;x):=\sum\limits_{k=0}^{\infty}
\frac{(a_1)_k,(a_2)_k,\ldots,(a_p)_k}{(b_1)_k,(b_2)_k,\ldots,(b_q)_k}\frac{x^k}{k!}\;\;\;\text{for}\;\;\;p,q=0,1,2,\ldots$$
is the generalized hypergeometric function \citep{Gradshteyn:Ryzhik:2007}
and $(a)_{k}:=  \Gamma(a+k)/\Gamma(a)$, for $k\in \N \cup \{0\} $, is
the Pochhammer symbol. Note that $\rho(h)=0$  implies pairwise independence,  as in the Gaussian case since (\ref{eq:weibull}) can be factorized in the product of two Weibull densities. Additionally, since ${}_2F_1\left(\cdot,\cdot,\cdot;0\right)=1$,
$\rho(h)=0$   implies $\rho_W(h)=0$  {that is}  if a  compactly supported correlation function \citep{Bevilacqua:20189} is used as underlying correlation model,
then also the correlation of the Weibull process is compactly supported.
 {This feature  is particularly  appealing  from the computational  point of view since  algorithms  for sparse matrices can be used
to handle the correlation matrix associated with $\rho_W$ (see Section 4.2).}

More important, it can be shown that some nice properties  such as stationarity, mean-square continuity, degrees of mean-square differentiability
and long-range dependence
can be inherited from the `parent' Gaussian process $Z$.
In particular, using the results in \citet[Section 2.4]{Stein:1999} linking  the behavior of the correlation   at the origin and the geometrical properties of the associated process, 
we can prove that
${W}$ is  mean square  continuous
if $Z$ is mean square  continuous
and it is
$k$-times mean-square differentiable if $Z$
is $k$-times mean-square differentiable.
Finally, it is trivial to see that the sample path continuity and differentiability   are inherited from the  `parent' Gaussian process.
As a consequence,  mean-square continuity and differentiability of the sample paths of the Weibull process can be modeled using suitable flexible parametric correlation functions as in the case of the Gaussian processes.

As an illustrative example, Figure \ref{fig:fcov-matern}   collects three simulations of  $W$ on a fine grid of $S=[0,1]^2$ with
 Mat{\'e}rn correlation function
$
\rho(h)=
{2^{1-\psi}}{\Gamma(\nu)}^{-1} \left ({\|h\|}/{\phi}
\right )^{\nu} {\cal K}_{\nu} \left ({\|h\|}/{\phi} \right )
,$
where $\phi,\nu>0$  and ${\cal K}_{a}$ is a modified Bessel function of the second kind of
order $a >0$.

We have considered three different parametrization for the smoothness parameter   $\nu=0.5, 1.5, 2.5$. Under this setting, the paths of the  'parent' Gaussian process is
$0, 1, 2-$times mean square differentiable,  respectively.
The values, $\phi= 0.067, 0.042, 0.034$, of the range parameter have been chosen in order to obtain  a practical range, i.e.   {the distance  at which the correlation $\rho(h)$,  equal to $0.05$
 equal to $0.2$}. 
Additionally
we have fixed  the shape  parameter of the Weibull distribution as  $\kappa=10, 3, 1$. The corresponding correlation functions $\rho_{W}(h)$ are plotted (from left to right) in the top panel of Figure  \ref{fig:fcov-matern} and the bottom panel reports the histograms of the observations.

It is apparent that the correlation $\rho_{W}(h)$ inherits  the  change of the  differentiability at the origin from  $\rho(h)$ when increasing $\nu$.
This changes have consequences on the geometrical properties  of the associated random processes. In fact the smoothness of the realizations (central panel of Figure \ref{fig:fcov-matern}) increase with $\nu$.
Note also the flexibility of the Weibull model when modeling positive data in the bottom panel 
of Figure  \ref{fig:fcov-matern}   since both  positive and negative skewness  can be achieved
with different values of $\kappa$. %

\begin{figure}
\begin{center}
\hspace{-1cm}
\begin{tabular}{ccc}
\hspace{-0.5cm}
\includegraphics[width=0.3\linewidth, height=0.2\textheight]{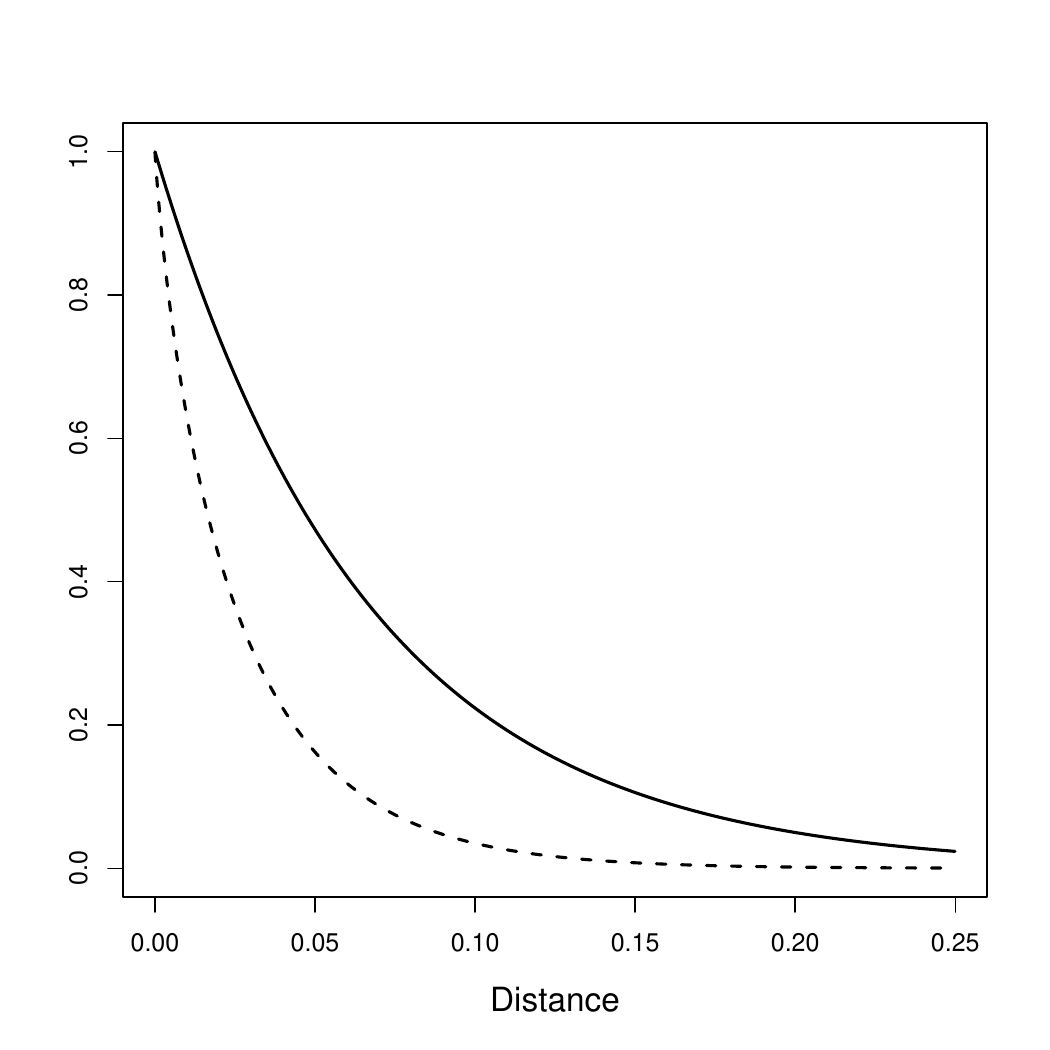} &
\hspace{-0.5cm}
\includegraphics[width=0.3\linewidth, height=0.2\textheight]{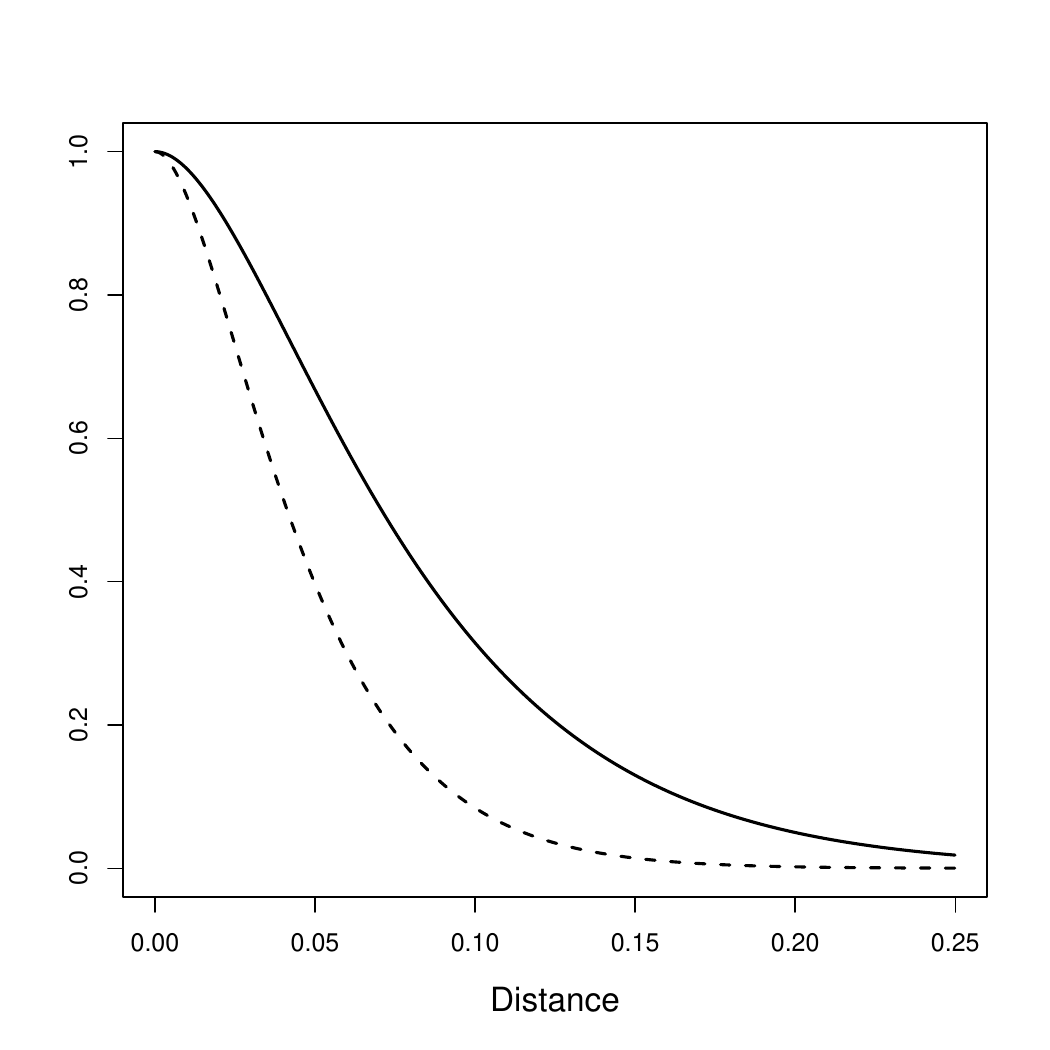}  &
\hspace{-0.5cm}
\includegraphics[width=0.3\linewidth, height=0.2\textheight]{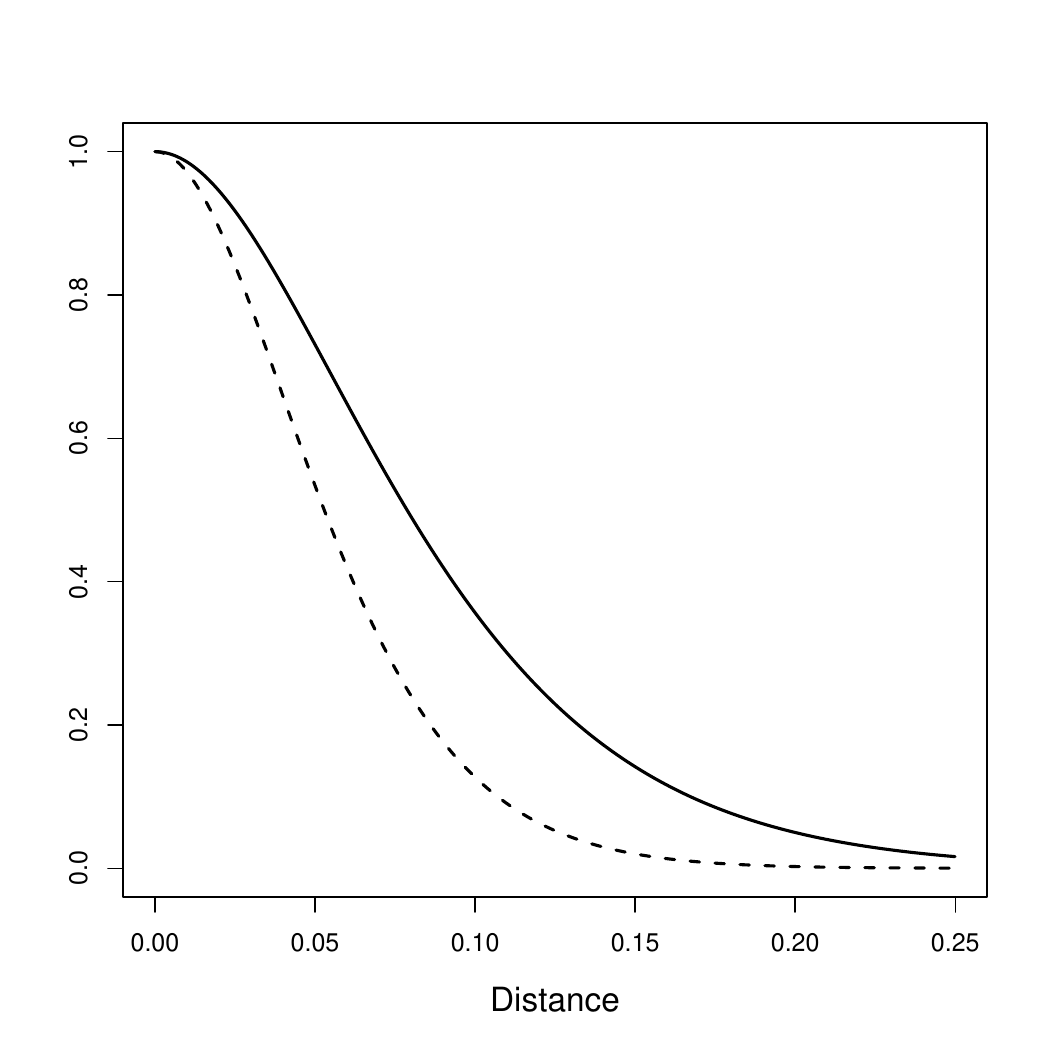}  \\
(a)&(b)&(c)\\
\hspace{-.3cm}
\includegraphics[width=0.35\linewidth, height=0.2\textheight]{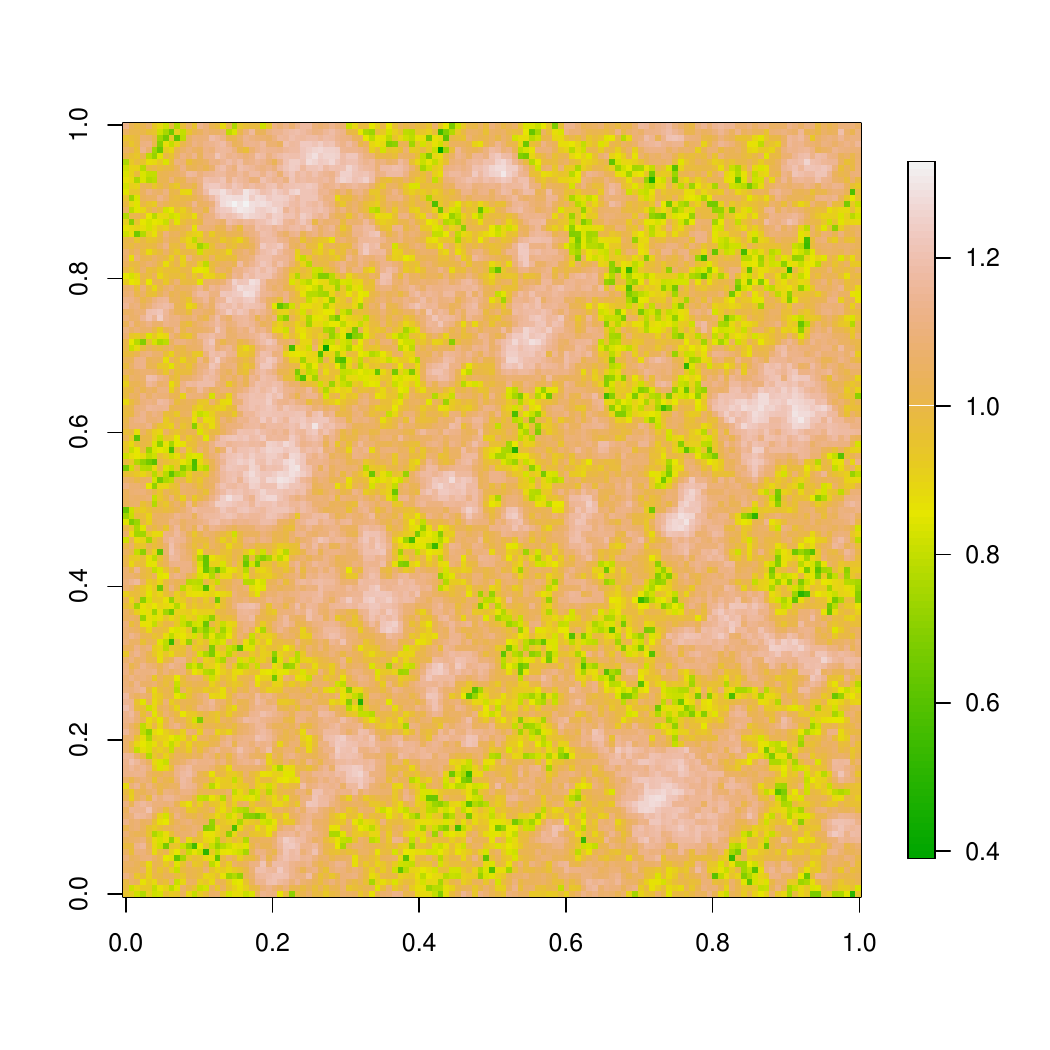}  &
\hspace{-.3cm}
\includegraphics[width=0.35\linewidth, height=0.2\textheight]{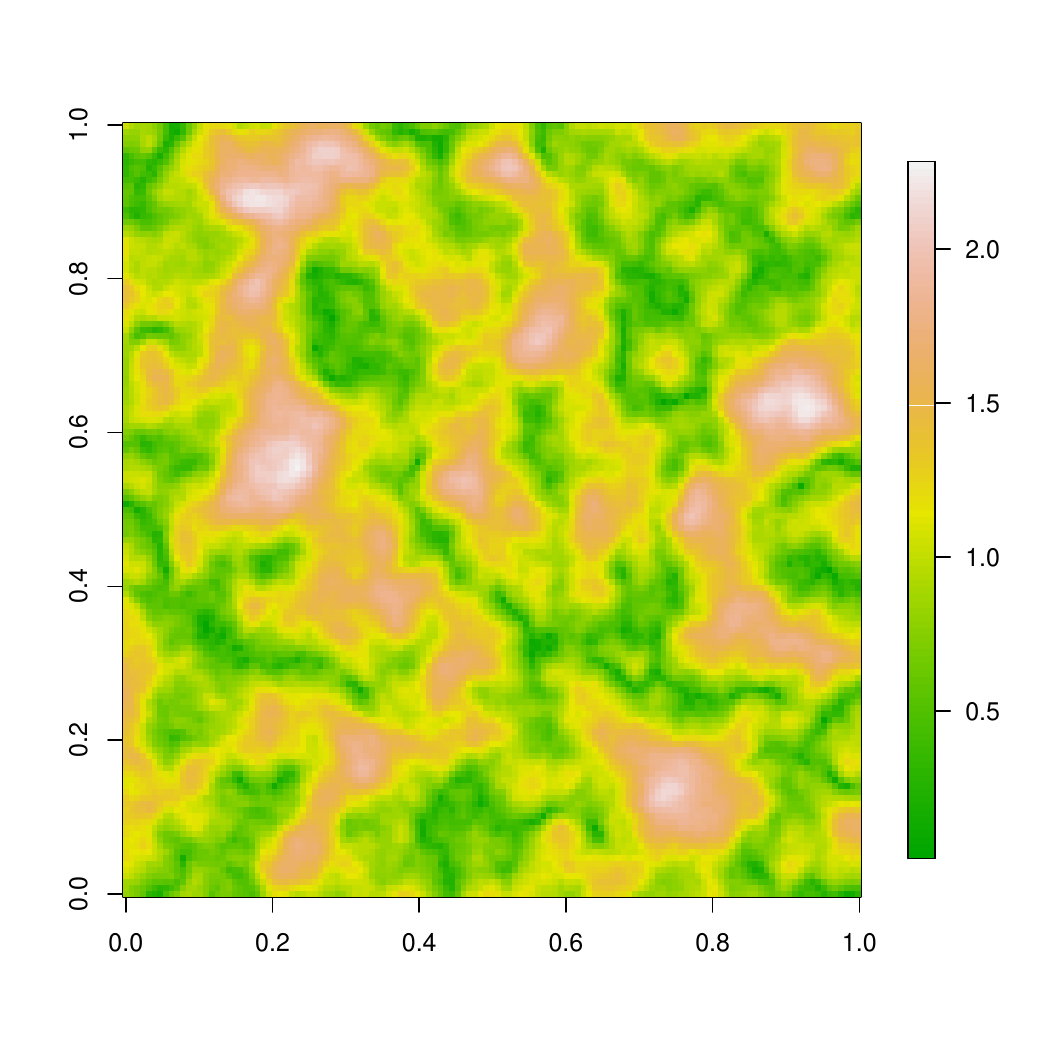} &
\hspace{-.3cm}
\includegraphics[width=0.35\linewidth, height=0.2\textheight]{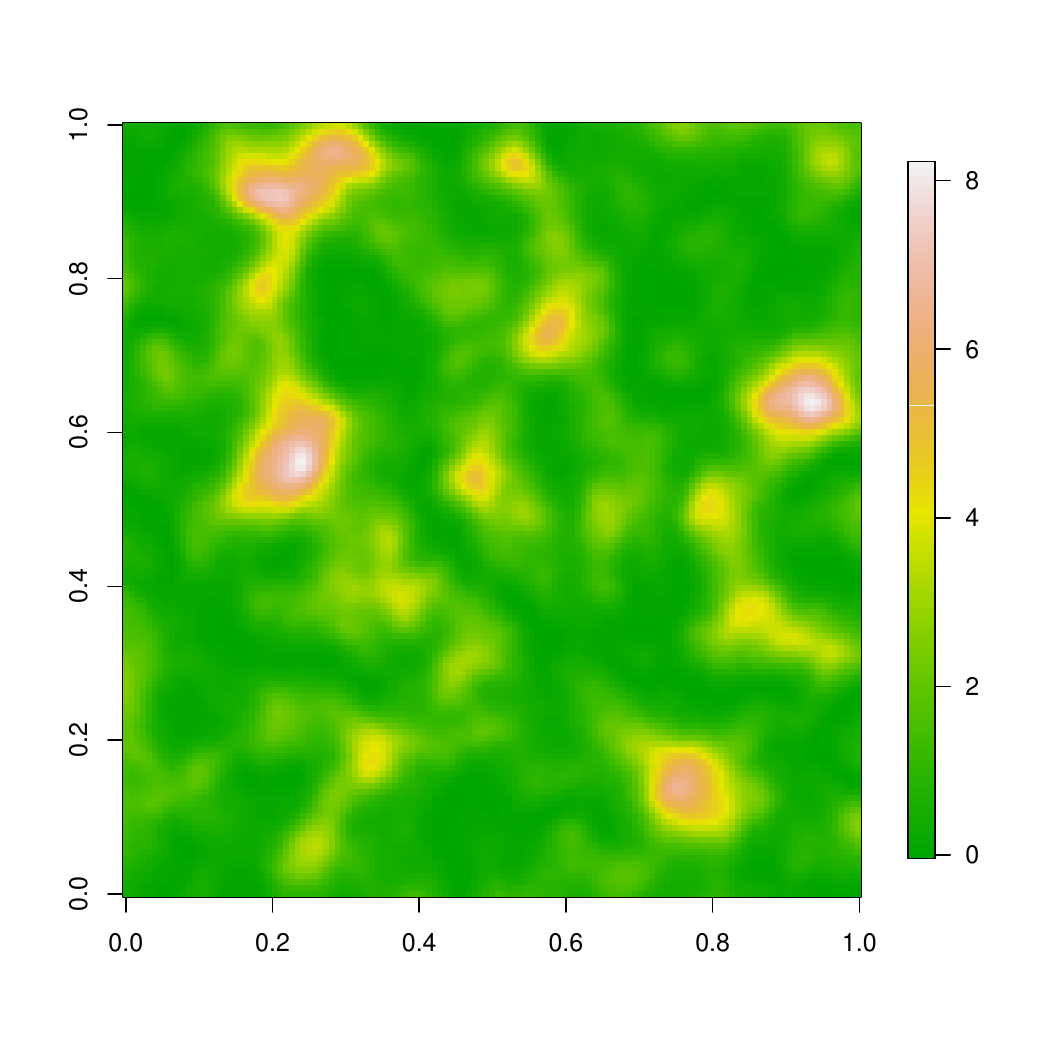}  \\
(d)&(e)&(f)\\
\hspace{-0.5cm}
\includegraphics[width=0.3\linewidth, height=0.2\textheight]{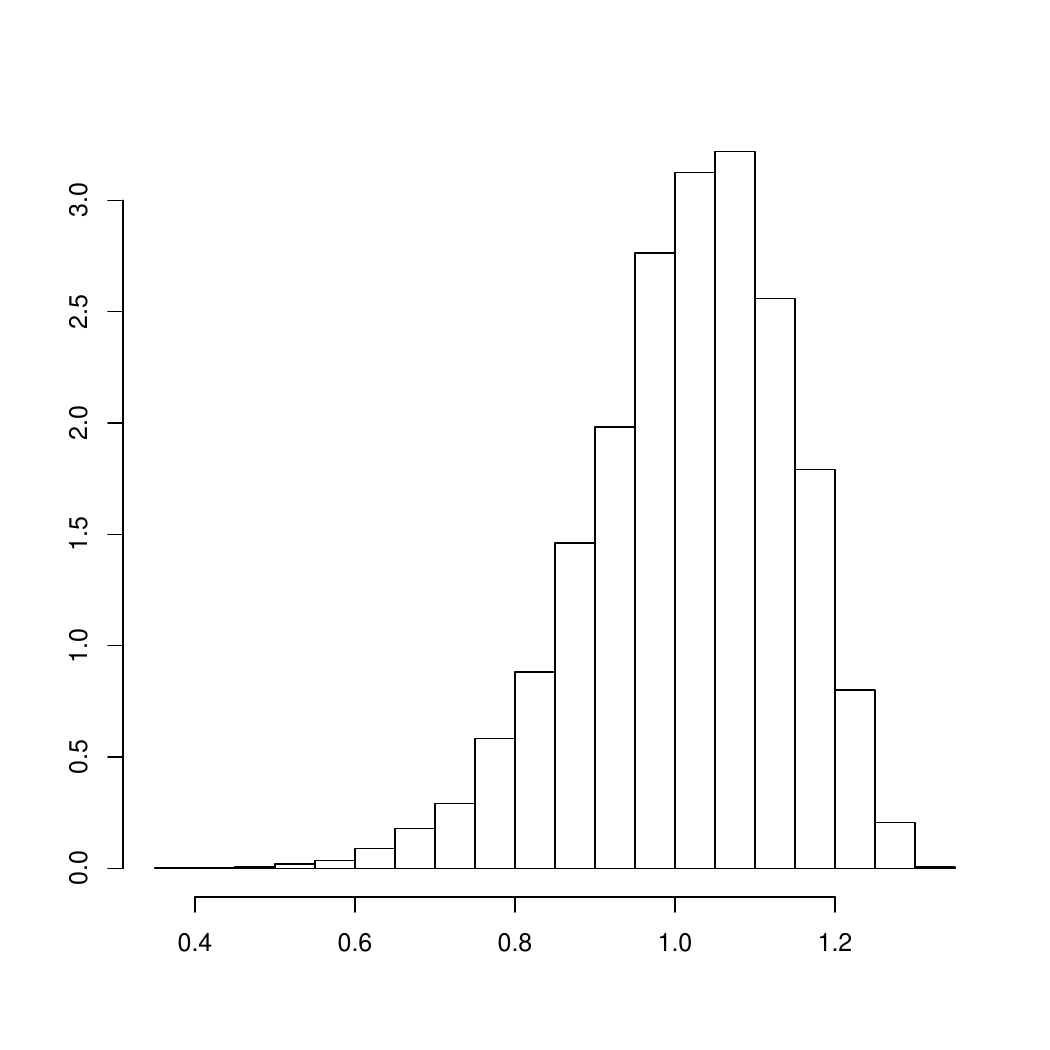} &
\hspace{-0.5cm}
\includegraphics[width=0.3\linewidth, height=0.2\textheight]{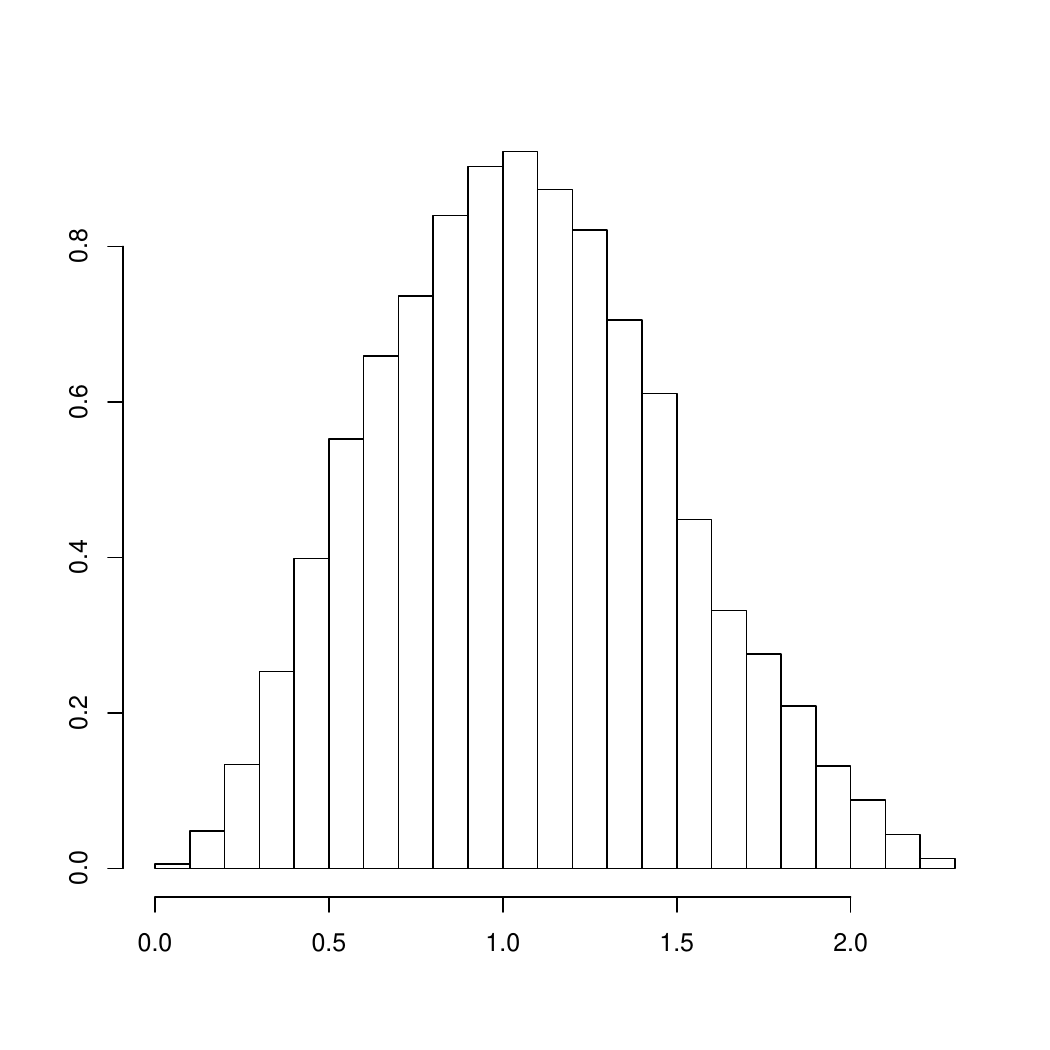}  &
\includegraphics[width=0.3\linewidth, height=0.2\textheight]{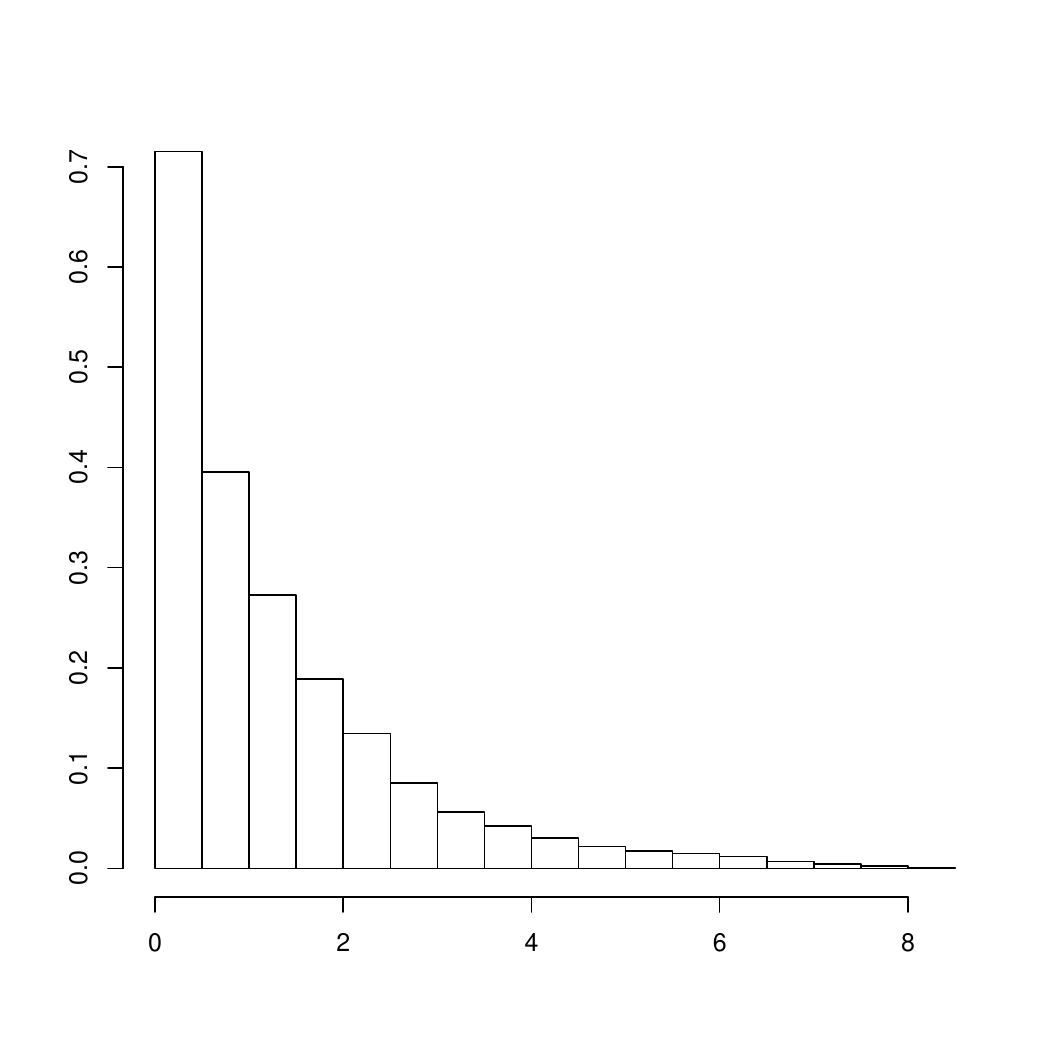}  \\
(g)&(h)&(i)
\end{tabular}
\end{center}
\caption{Top:
comparison between $\rho(h)$, the  correlation function of the `parent' Gaussian process (a Mat\'ern correlation function with smoothness parameter $\nu=0.5, 1.5, 2.5$
and practical range  approximately equal to $0.2$),  with the associated correlation of the Weibull model  $\rho_W(h)$  (dashed line) 
with shape parameter $\kappa$ for  $(\kappa,\nu)=(10,0.5) ,   (3,1.5),  (1,2.5)$,  from  left to right.
Center: three realizations of the Weibull model $W$ under the setting (a),(b), and (c).
Bottom: histograms  of the realizations in (d),(e), and (f).}
 \label{fig:fcov-matern}
\end{figure}
Finally, a non-stationary version of $W$
 can be easily obtained trough a multiplicative model:
\begin{equation}\label{def:reg-weibull}
{Y}({s}):=
\mu(s) W({s}),
\end{equation}
where $\mu(s)>0$ is a non random function that specify the mean of $Y$ $i.e.$
$\E(Y(s))=\mu(s)$ and affects its variance
$\var({Y}({s}))=\mu(s)^2(\Gamma\left(1+2/ \kappa\right)\nu^2(\kappa)-1)$.

A useful parametric specification for $\mu(s)$ is given  through log-linear link function 
 $$\log(\mu({s})) =\beta_0+\beta_1 v_1(s)+\cdots+\beta_{p}v_{p}(s)$$
where $v_i(s)$, $i=1,\ldots,p$, are   covariates and  $\beta=(\beta_0,\ldots,\beta_{p})^\top$ is a vector of regression parameters but other types of parametric or nonparametric functions can be considered.

Finally, note that  given the observations $y(s_1),\ldots,y(s_n)$
and   an estimation of the mean function $\hat{\mu}(s_i)$, the estimated residuals 
$\widehat{w}(s_i)={y}({s_i})/\hat{\mu}(s_i)$
can be treated as a realization of the stationary Weibull process $W$ with marginal distribution $\mbox{Weibull}( \kappa, \nu(\kappa))$ and  correlation  $\rho_{W}(h)$. This can be used  to check 
the agreement between  the distribution of the residuals and the estimated theoretical marginal distribution or to check
 the agreement between the theoretical estimated semi-variogram model obtained from (\ref{ccc}) and its empirical counterpart (see Section 6).

\section{Estimation and prediction}\label{sec:est-pred}
\subsection{Pairwise likelihood inference}
Suppose that we have observed $y_1,\ldots,y_n$  at the locations $s_1,\ldots,s_n$ and
let $\theta$ be the vector of unknown parameters for the Weibull random process  (\ref{def:reg-weibull}).
The  evaluation of the full likelihood for $\theta$ is impracticable 
since, as outlined in Section 2, the multivariate density 
can be derived only in some special cases. A possible  alternative \citep{Lindsay:1988,Varin:Reid:Firth:2011}  combines  the bivariate distributions of  all possible distinct pairs of observations $(y_i,y_j)$.
 The weighted pairwise likelihood (WPL) function
 is given by
\begin{equation}
\label{pln}
\operatorname{pl}(\theta):=\sum_{i=1}^n\sum_{j>i}^n\log f (y_i,y_{j};\theta) {c_{ij}}
\end{equation}
where
$f (y_i,y_j;\theta)$ is  the bivariate
densities of (\ref{def:reg-weibull})  and
$c_{ij}$ are non-negative weights.  The choice of  cut-off weights, namely $c_{ij}= 1$ if
$\|s_i-s_j\| \le \Delta$, and $0$ otherwise, for a positive value of $\Delta$, can be  motivated by its simplicity and by observing that that dependence  between observations which are distant  is weak.  Therefore, the use of all  pairs may  skew the information confined in pairs of near observations
\citep{Davis:Yau:2011,Bevilacqua:Gaetan:2015}.

Under the increasing domain asymptotics framework \citep{Cressie:1993} and arguing as in  \citet{Bevilacqua:Gaetan:2015},
it can be shown that the  maximum weighted pairwise likelihood (MWPL) estimator $
\widehat{\theta}:=\operatorname{argmax}_\theta\, \operatorname{pl}(\theta)
$
is consistent and asymptotically Gaussian.
The  asymptotic covariance matrix of the estimator is given by
the inverse of the Godambe information
$$\mathcal{G}_n(\theta):=\mathcal{H}_n(\theta)^\top\mathcal{J}_n(\theta)^{-1}\mathcal{H}_n(\theta),
$$
where
$\mathcal{H}_n(\theta):=\E[-\nabla^2 \operatorname{pl}(\theta)]$ and $\mathcal{J}_n(\theta):={\var}[\nabla \operatorname{pl}(\theta)]$.

The matrix $\mathcal{H}_n(\theta)$ can be estimated by $\widehat{\mathcal{H}}=
-\nabla^2 \operatorname{pl}(\widehat{\theta})$ and  the estimate  $\widehat{
	\mathcal{J}}$ of $\mathcal{J}_n(\theta)$  can be  calculated with a sub-sampling
technique \citep{Heagerty:Lele:1998,Bevilacqua:Gaetan:Mateu:Porcu:2012}. 
Additionally, model selection can be performed by considering  the pairwise likelihood information criterion (PLIC) \citep{Varin:Vidoni:2005}
$$
\mbox{PLIC}:= - 2\operatorname{pl}(\hat{\theta})  + 2\mathrm{tr}(\widehat{\mathcal{J}}
\widehat{\mathcal{H}}^{-1})
$$
which is the composite likelihood version of
the Akaike information
criterion (AIC).

\subsection{Linear prediction}\label{sec:numerical_examples}

The lack of  workable multivariate densities forestalls the use of the conditional distributions for the prediction.
Therefore we choose a sub optimal solution, i.e. a linear predictor  for the random variable $Y(s_0)$ at some unobserved location ${s}_0$ based on the data at locations $s_1,\ldots,s_n$ ,  following a suggestion in  \citet{Bellier:Monestiez:Guinet:2010} and \citet{DeOliveira:2014}.

The predictor for the non-stationary Weibull process  is given by
\begin{equation}\label{pred-sk}
\widehat{Y}(s_0):=\mu(s_0)\left\{1+\sum_{i=1}^n\lambda_i (W(s_i)-1)\right\},
\end{equation}
where  $W(s_i)={Y(s_i)}/{\mu(s_i)}$.
It is a  linear predictor and  unbiased predictor of  $Y(s_0)$ for any
vector of weights $\lambda=(\lambda_1,\ldots,\lambda_n)^\prime$. 
The vector of weights  $\lambda=(\lambda_1,\ldots,\lambda_n)^\prime$  is set by minimizing the mean square error $\E[Y(s_0)-\hat{Y}(s_0)]^2$ with respect to  $\lambda$.
 {Note that this is a classical geostatistical  approach  applied to a multiplicative  model instead of the classical additive model.}
It turns out that the  solution for the predictor is given by the equations of the simple %
kriging,  \citep[Chapter 3]{Cressie:1993} i.e.
$\lambda =  C^{-1}_W\, c_W(s_0)$
and  the associated mean square prediction error is given by
$$\var( \widehat{Y}(s_0)):=\mu^2(s_0)\sigma^2_W\left\{1-  c_W(s_0)^\prime C_W^{-1} c_W(s_0)\right\},$$
where $\sigma^2_W:=(\Gamma\left(1+2/ \kappa\right)\nu^2(\kappa)-1)$, 
$c_W(s_0)= (\rho_W(s_0-s_i),\cdots,\rho_W(s_0-s_n))^\prime$  
and $C_W$ is the $n\times n$ correlation matrix whose $(i,j)$th element is $\rho_W(s_i-s_j)$
with $\rho_W(h)$  given in  (\ref{ccc}). 

In practice the predictor cannot be evaluated since  $\mu(s)$ and $\rho_W(h)$ are
unknown. For this reason we suggest to use a plug-in estimate for
$\mu(s)$ and $\rho_W(h)$ using the pairwise likelihood estimates.

\section{Simulation results}\label{sec:numerical-results}

 {In this section we investigate, through some numerical experiments,  the relative  efficiency of the  MWPL estimator with respect to the maximum  likelihood (ML) estimator
and the relative  efficiency of the  linear predictor (\ref{pred-sk})  with respect to the optimal predictor,
under a specific setting of simulation where the comparisons can be explicitly performed.}
This specific setting is when the process is defined on $\R$ and the underlying correlation function is exponential.
Even though this setting may seem artificial,
the simulation study gives an idea of the relative efficiency of the MWPL estimation method and the proposed linear predictor under more general settings. 

We have considered a non-stationary Weibull model (\ref{def:reg-weibull}) observed at   $150$ locations of a regular grid  $0=s_1<s_2<\cdots< s_{150}=1$
where the `parent'  Gaussian random process has exponential correlation function
$\rho_{i,j}:=\rho(s_i-s_j)= \exp(-|s_i-s_j|/\phi)$. 

In this case, the multivariate density function   associated with the Weibull process is easily obtained from    (\ref{gammafd1}), namely
\begin{equation}\label{qqq}
f_{Y}(y_1,\ldots,y_n)= \left\{\frac{\kappa}{\nu(\kappa)^\kappa}\right\}^n
f_{X_2}(x_1,\ldots x_n) \prod_{i=1}^n \frac{y_i}{\mu_i} \,
\end{equation}
where $x_i:=\{y_i/(\nu(\kappa)\mu_i)\}^\kappa$, $\mu_i:=\mu(s_i)$ and $f_{X_2}$ can be obtained from (\ref{gammafd1}).
Therefore this  setup allows a comparison of the MWPL and  ML  estimation methods.

We set $\mu(s)=\exp\{\beta_0+\beta_1 v_1(s)\}$ where $v_1(s)$ is a value from the $(0,1)$-uniform distribution and $\beta_0=0.25$ and $\beta_1=-0.15$. Three choices  of  the  shape parameter  $\kappa=1, 3, 10$ are coupled with three values of the range parameter $\phi=a/3$, $a=0.1, 0.2,0.3$.

 {We simulate   1,\,000 realizations from each model setting and for each realization, we calculate $\theta^a_k$, $k=1,\ldots,1000$, $a=ML,MWPL$, estimates of $\theta=(\beta_0,\beta_1,\phi,\kappa)^\prime$.
We set $\Delta$ equal to the minimum distance among the points in \eqref{pln} and we use the true value of the parameters as starting value for the Broyden-Fletcher-Goldfarb-Shanno (BFGS) algorithm 
implemented in the \texttt{optim} function of \texttt{R} software \citep{RR}.}

Table \ref{tab:relative} reports the relative efficiency of the  MWPL estimates with respect to the ML estimates  {for each parameter in terms of mean squared error. Additionally, as an overall measure of relative efficiency for the multi-parameter case we have considered}
$$
\operatorname{RE}=\left(\frac{\operatorname{det}[F^{MWPL}]}{\operatorname{det}[F^{ML}]}\right)^{1/p},
$$
 {where  $p = 4$ is the number of unknown parameters in $\theta$ and the matrix $F^a$  is the  sample mean squared error matrix 
$F^a=1000^{-1}\sum_{k=1}^{1000}\left(\hat\theta_{k}^a-\theta\right)\left(\hat\theta_k^a-\theta\right)^\prime$. 
In this experiment, using the WPL instead of the  likelihood function, we loose about 13\% of the overall efficiency in the worst case which is an encouraging result.
It is interesting to note that the relative efficiency   of each parameter is different, but  only the  relative efficiency of the  shape parameter $\kappa$ is  affected when we consider different strengths of the spatial dependence,  i.e. different values of $\phi$. }

\begin{table}[!hbtp]
\begin{center}
\scalebox{0.8}{
\begin{tabular}{|c|c|c|c|c|}
\hline
 $\kappa$&  & $\phi=0.1/3$ & $\phi=0.2/3$ & $\phi=0.3/3$ \\ \hline
\multicolumn{1}{|c|}{\multirow{5}{*}{$1$}} & $\beta_0$ & $0.964$ & $0.967$ & $0.956$ \\
\multicolumn{1}{|c|}{} & $\beta_1$ & $0.862$ & $0.860$ & $0.874$ \\
\multicolumn{1}{|c|}{} & $\phi$ & $1.045$ & $1.037$ & $1.034$ \\
\multicolumn{1}{|c|}{} & $\kappa$ & $0.885$ & $0.703$ & $0.550$ \\ \cline{2-5}
\multicolumn{1}{|c|}{} & RE & $0.954$ & $0.913$ & $0.884$ \\ \hline
\multicolumn{1}{|c|}{\multirow{5}{*}{$3$}} & $\beta_0$ & $0.953$ & $0.947$ & $0.930$ \\
\multicolumn{1}{|c|}{} & $\beta_1$ & 0.862 & 0.860 & $0.874$ \\
\multicolumn{1}{|c|}{} & $\phi$ & $1.045$ & $1.036$ & $1.034$ \\
\multicolumn{1}{|c|}{} & $\kappa$ & $0.885$ & $0.703$ & $0.550$ \\ \cline{2-5}
\multicolumn{1}{|c|}{} & RE & $0.955$ & $0.914$ & $0.886$ \\ \hline
\multicolumn{1}{|c|}{\multirow{5}{*}{$10$}} & $\beta_0$ & $0.947$ & $0.933$ & $0.911$ \\
\multicolumn{1}{|c|}{} & $\beta_1$ & $0.862$ & $0.860$ & $0.874$ \\
\multicolumn{1}{|c|}{} & $\phi$ & $1.044$ & $1.036$ & $1.034$ \\
\multicolumn{1}{|c|}{} & $\kappa$ & $0.885$ & $0.703$ & $0.550$ \\ \cline{2-5}
\multicolumn{1}{|c|}{} & RE & $0.955$ & $0.914$ & $0.886$ \\ \hline
\end{tabular}
}
\end{center}
\caption{Mean squared error relative efficiency  for each parameter and overall relative efficiency  (RE) of MWPL vs ML. }
\label{tab:relative}
\end{table}

 We modify slightly  our  example to illustrate  the quality  of the linear predictor (\ref{pred-sk}) in terms of  the mean squared prediction error (MSPE).  
Suppose that we have observed $Y(s_1)=y_1,\ldots, Y(s_n)=y_n$ 
and we want to predict $Y(s_{n+1})$ with $s_{n+1} > s_n$. In such case the conditional expectation of $Y(s_{n+1})$, i.e. the  predictor the minimizes the MSPE, 
can be derived  in closed form (see Appendix), namely 
\begin{equation*}
\begin{split}
Y^*(s_ {n+1}):=&\Gamma\left(\frac{1}{\kappa}+1\right)(1-\rho^2_{n,n+1})^{1/\kappa}\mu_{n+1}\nu(\kappa)\\
&\times
\exp\left\{-\frac{[y_n/(\mu_n \nu(\kappa))]^\kappa}{(1-\rho^2_{n-1,n})}    \left[
\frac{   (1-\rho^2_{n-1,n} \rho^2_{n,n+1})}{(1-\rho^2_{n,n+1})}-1 \right]\right\}\\
&\times {}_1F_1\left(\frac{1}{\kappa}+1;1;\frac{
[y_n/(\mu_n \nu(\kappa))]^\kappa}{(1-\rho^2_{n,n+1})}\rho^2_{n,n+1}\right).
\end{split}
\end{equation*}

Having collected
$n=21$ observations at locations  $s_1=0,s_2=0.05, \ldots,  s_{n}=1$, we predict the random variable $Y(s_{n+1})$ at $s_{n+1}=1.05$   by means of the optimal predictor $Y^*(s_{n+1})$, and the linear predictor $\widehat{Y}(s_{n+1})$ as in (\ref{pred-sk}).

We   simulate  1,\,000 realizations from  the  stationary Weibull model, i.e. $\mu(s_i)=1$, with the same dependence structure as before. Then we compute the average of the squared prediction errors $[Y(s_{21}) -{Y}^*(s_{21})]^2$ and $[Y(s_{21}) -\widehat{Y}(s_{21})]^2$ and their ratio.
Table \ref{tab:relativeprediction}  shows the ratio  between the  linear and the optimal predictor.
This ratio deteriorates when the strength of the dependence increases as expected, but the loss of the relative efficiency does not exceed thirty-two percent.

\begin{table}[t]
\begin{center}
\scalebox{0.8}{
\begin{tabular}{c|c|c|c|}
\cline{2-4}
& $\phi=0.1/3$ & $\phi=0.2/3$ & $\phi=0.3/3$ \\ \hline
\multicolumn{1}{|c|}{$\kappa =1$} & $0.953$      & $0.805$      & $0.687$      \\ \hline
\multicolumn{1}{|c|}{$\kappa =3$} & $0.960$      & $0.825$      & $0.721$      \\ \hline
\multicolumn{1}{|c|}{$\kappa =10$} & $0.967$      & $0.851$      & $0.764$      \\ \hline
\end{tabular}
}
\end{center}
\caption{Relative efficiency of the  linear predictor  versus the optimal predictor  for a stationary Weibull model defined on $S=[0,1]$ with underlying exponential correlation $\rho(h)=\exp\{-|h|/\phi\}$.}
\label{tab:relativeprediction}
\end{table}

\section{Wind speed data example }\label{sec:real-data}

Our motivating example is 
a dataset of  daily 
average wind speeds  from a network of meteorological stations in the Netherlands. The dataset %
 is  stored on the website of the KNMI Data Centre (\texttt{https://data.knmi.nl/about}) and its access  is  provided under the OpenData policy of the Dutch government.

Among the fifty stations in the dataset, we extracted thirty stations (Figure \ref{fig:map-ned}-a) that do not contain missing data in the period from 01/01/2000 to 31/12/2008 .

Figures \ref{fig:map-ned}-(b,c,d) show  the time series plots of  daily mean wind speeds at four different locations  (Cabauw, Nieuw Beerta, Hoek Van Holland and Rotterdam) in 2000-2004.
The seasonal trend is clearly recognizable and the heteroscedasticity seems related to this trend. Furthermore  if we consider the wind speed box-plots for each station (Figure \ref{fig:boxplots}-a), it is clear that the distribution also depends on the location.
   To avoid a complicated spatial trend specification, we transform 
$Y(s,t)$,   the observation of location $s$ and time $t$, 
to $\widetilde{Y}(s,t)=Y(s,t)/a(s)$ where $a(s)$ is the average of the observations at site $s$. The transformation seems to have an effect of reducing the differences in distribution, see (Figure \ref{fig:boxplots}-b). 

We specify  a multiplicative model for the transformed data, namely
\begin{equation}\label{eq:mult}
\widetilde{Y}(s,t)=\mu(t)E(s,t),
\end{equation}
in which we conveys the  seasonal pattern in the deterministic positive function $\mu(t)$ and $E=\{E(s,t)\}$ is a stationary positive process with unit mean.
In particular we specify  a harmonic model for the temporal trend, i.e.
\begin{eqnarray}\label{eq:harmonics}
\log\mu(t) &=&\beta_0+ \sum_{k=1}^q \left\{\beta_{1,k}\cos\left(\frac{2\pi k t}{P}\right)+\beta_{2,k}\sin\left(\frac{2\pi k t}{P}\right)\right\}
\end{eqnarray}
where  we  set $P=365.25$ days to  handle leap years.

In the  sequel we want  to compare  two specifications of $E$, namely
the proposed Weibull model $E(s,t)=W(s,t)$ and 
a log-Gaussian model $E(s,t)=\exp(\sigma Z(s,t)-\sigma^2/2)$, $\sigma >0$ where $Z$ is a standard space-time Gaussian process.

 We   first get a preliminary estimate of the seasonal effect $\mu(t)$ assuming space-time independence and by using least squares and regressing $q=4$ annual harmonics on the 
  logarithm of the observations
\begin{equation}\label{eq:prel-reg}
\log\widetilde{Y}(s,t)=\beta_0+ \sum_{k=1}^4 \left\{\beta_{1,k}\cos\left(\frac{2\pi k t}{P}\right)+\beta_{2,k}\sin\left(\frac{2\pi k t}{P}\right)\right\}+\varepsilon(s,t)
\end{equation}
with  $\E(\varepsilon(s,t))=0$ and $\var(\varepsilon(s,t))=\sigma^2_\varepsilon<\infty$.  Under  the Weibull marginal distribution for $E(s,t)$ we 
identify  
$\beta_0$ with $\beta_0+\log(\nu(\kappa))- {\gamma}/{\kappa}$  since   $-\log W(s,t)$ is a Gumbel random variable with mean $-\log \nu(\kappa)+ {\gamma}/{\kappa}$ 
and $\gamma \approx 0.5772$ is the Euler–Mascheroni constant. 
Instead,  under  the log-Gaussian marginal distribution for $E(s,t)$ we 
identify  
$\beta_0$ with $\beta_0-\sigma^2/2$.

 {We have used  the values $e(s,t)=\exp(\widehat{\varepsilon}(s,t))$, where   $\widehat{\varepsilon}(s,t)$ are the   estimated residuals  of the fitted  regression model (\ref{eq:prel-reg}), 
for getting more insight about 
soundness of the  model specification (\ref{eq:mult}) for the wind speed data.
Since parameter $\mu(t)$ controls at the same time the expectation and the variance of $\tilde{Y}(s,t)$, it is expected that the residuals $e(s,t)$ are homoschedastic. This is confirmed by grouping $e(s,t)$ by month and looking at the corresponding boxplots (Figure \ref{fig:res-1}-(a)). Moreover comparing the overall qq-plots of $e(s,t)$ (Figure \ref{fig:res-1}-(b))
 there is convincing evidence that a Weibull distribution is more appropriate  with respect to the log-Gaussian one.}

In addition, if  we transform the residuals of each location to the normal scores by means of the empirical transform,
 the scatter-plots of the normal scores of Rotterdam station versus the  normal scores of three other stations  {(Figure \ref{fig:res})}  point out that there is more dependence in the upper corner, i.e. the lack of symmetry. This implies that  the Weibull model seems more appropriate  for modeling the pairwise dependence  with respect  to a
 log-Gaussian model since its copula  is reflection symmetric.

Finally, the spatial and temporal marginal  empirical semi-variograms  of the residuals %
 exhibit a strong and long decay dependence for the spatial margin and a weak dependence for the temporal margin.
This suggests the use of the following   space-time correlation   \citep{Porcu:Bevilacqua:Genton:2019}:

\begin{equation}\label{eq:space-time-corr}
\rho(h,u)=
\frac{1}{(1+\|h\|/\phi_S)^{2.5}}  \left (1 -\frac{|u|}{\phi_T(1+\|h\|/\phi_{S})^{{-\phi_{ST}}}} \right)_+^{3.5}, \end{equation}
with  $\phi_S> 0$, $\phi_T> 0$, $0\leq \phi_{ST} \leq 1$. %
When the space-time interaction parameter $\phi_{ST}$ is zero,  then the space-time correlation is simply the product of a spatial Cauchy correlation function and  a temporal  Wendland correlation function \citep{Bevilacqua:20189}, i.e.  a separable model for the underlying spatio temporal Gaussian process. However, from (\ref{ccc})
it is apparent that  separability  is not inherited for the  Weibull and Log-Gaussian models. In this application, we have considered three different degree of space-time interaction  by fixing 
$\phi_{ST}=0,  0.5, 1$.

\begin{figure}
\begin{center}
\begin{tabular}{cc}
\includegraphics[width=0.50\linewidth, height=0.30\textheight]{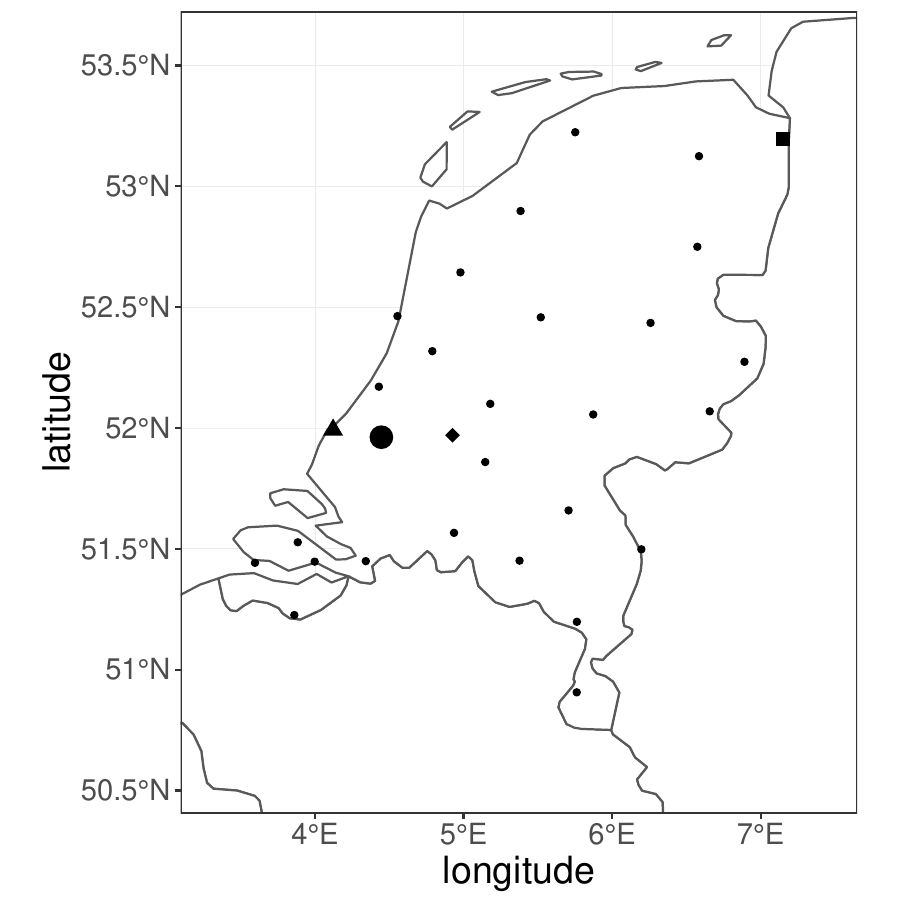}
&
\includegraphics[width=0.45\linewidth, height=0.30\textheight]{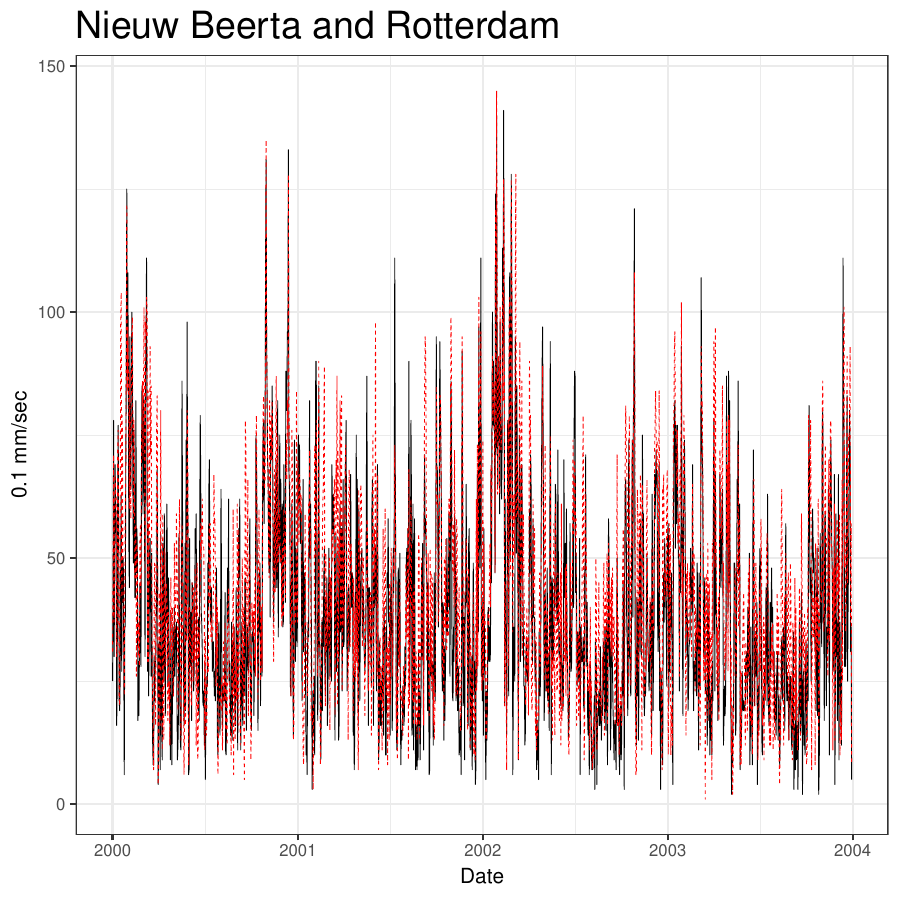}
\\
(a) & (b)
\\
\includegraphics[width=0.45\linewidth, height=0.30\textheight]{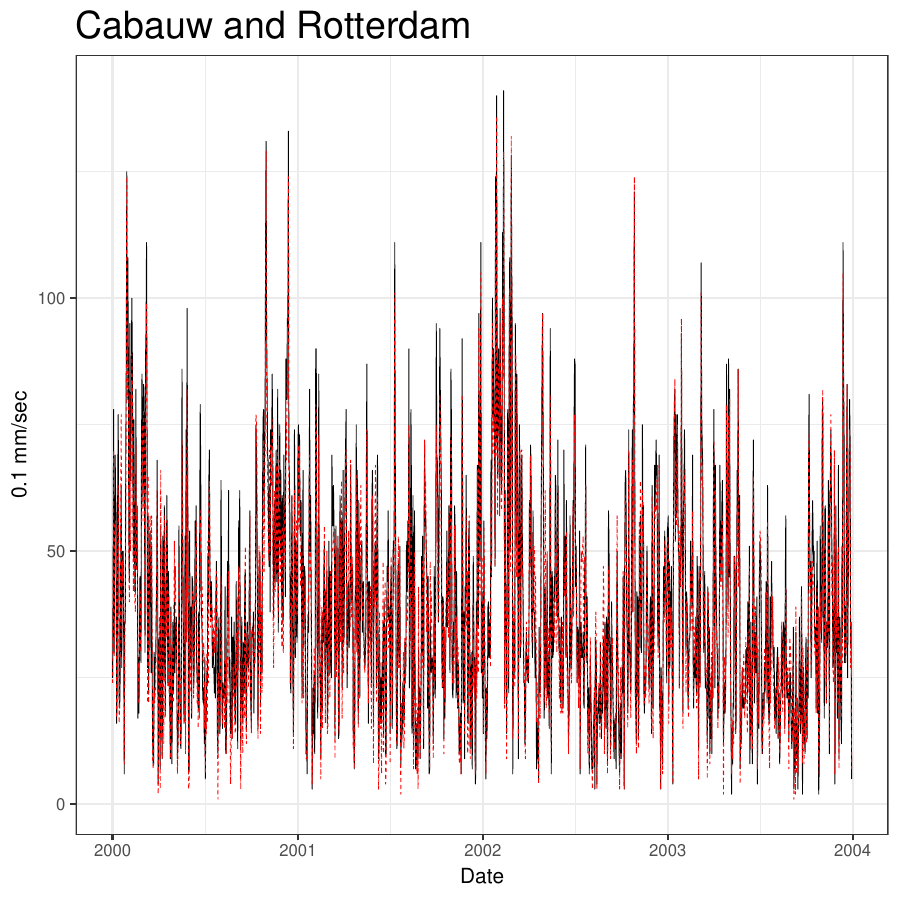}  & 
\includegraphics[width=0.45\linewidth, height=0.30\textheight]{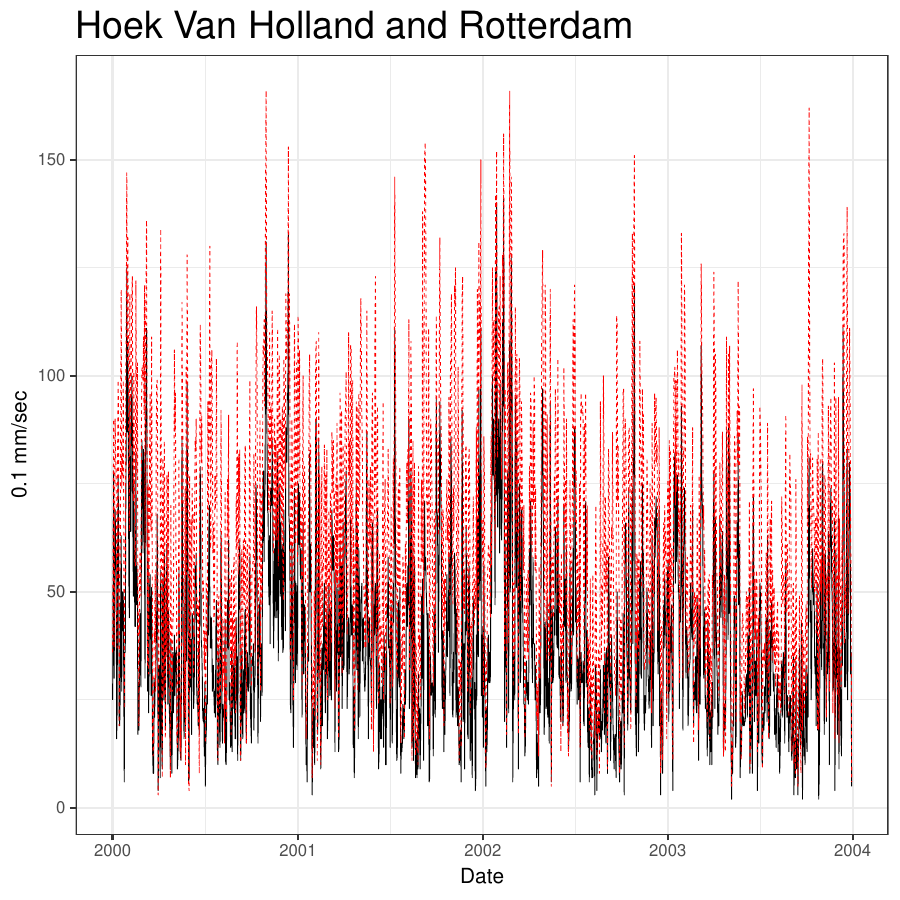}
\\
(c) & (d)
\end{tabular}
\end{center}
\caption{Wind speed data of Netherlands. (a) Map of the meteorological stations selected for our case study. Symbols     $\blacktriangle$, $\blacklozenge$, $\blacksquare$, \CIRCLE \,  correspond to    Cabauw, Hoek Van Holland, Nieuw Beerta and Rotterdam stations; (b-c-d) Time series plots (black lines<) of the daily  wind speed data (01/01/2000-31/12/2004) at Cabauw, Hoek Van Holland and Nieuw Beerta stations versus Rotterdam stations (red line).}
\label{fig:map-ned}
\end{figure}

\begin{figure}
\begin{center}
\begin{tabular}{cc}
\includegraphics[width=0.45\linewidth, height=0.3\textheight]
{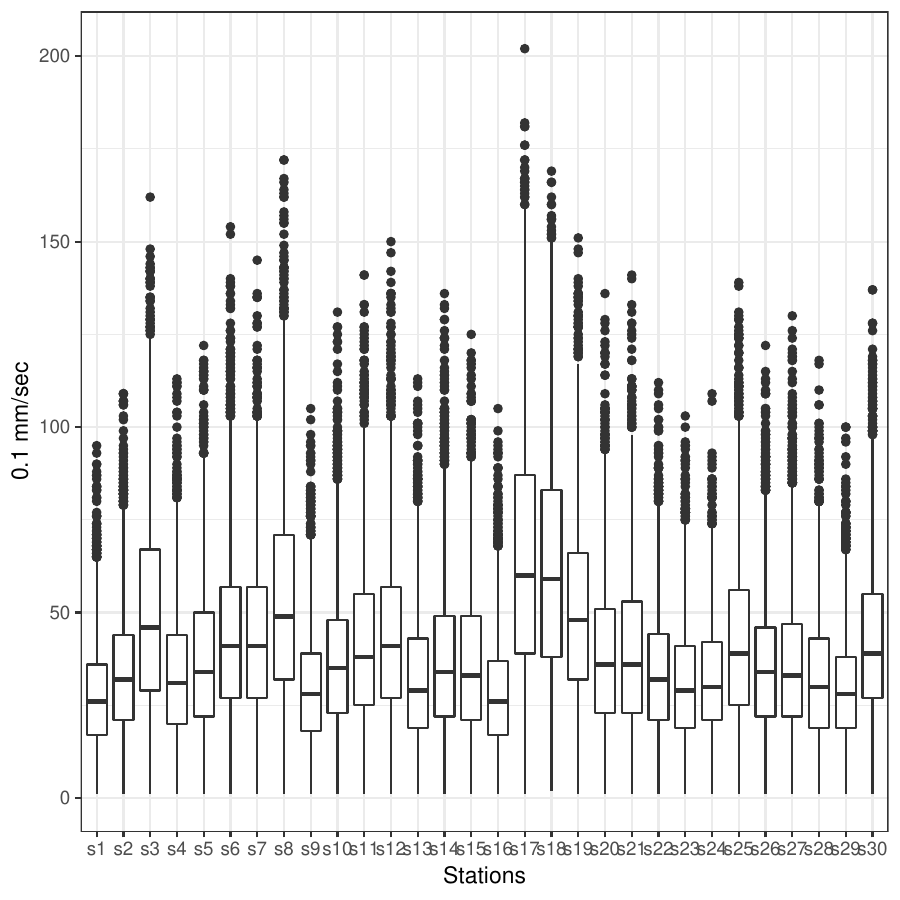} &
\includegraphics[width=0.45\linewidth, height=0.3\textheight]
{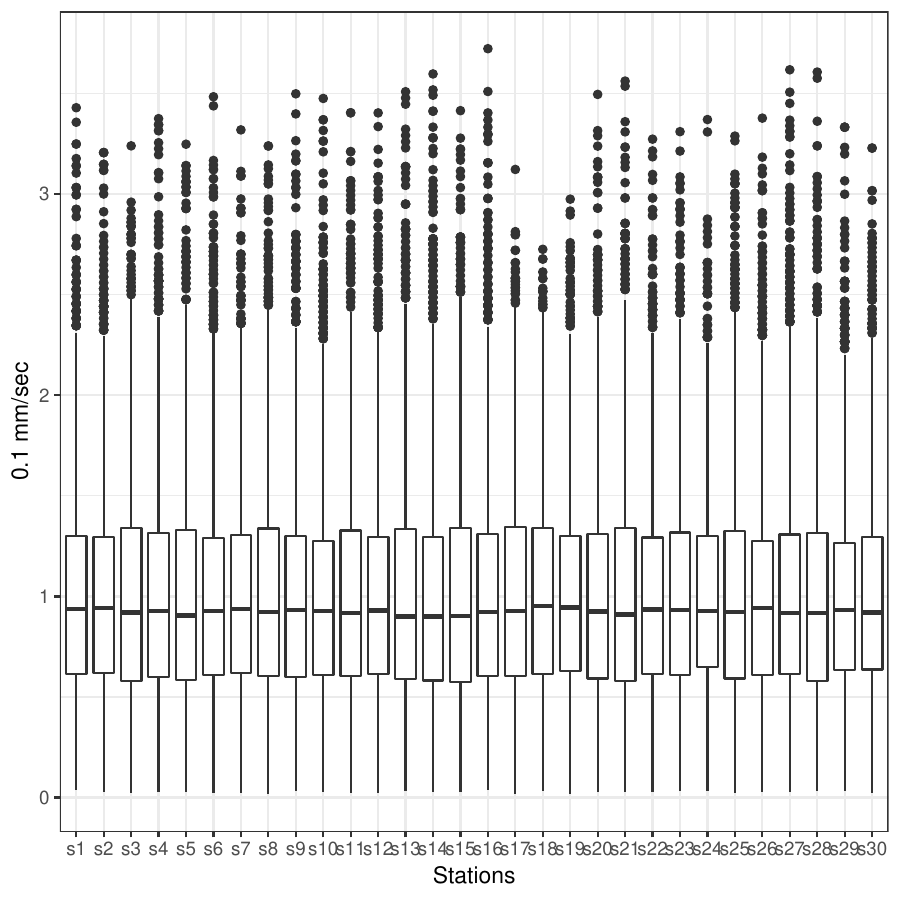}
\\
(a) & (b)
\end{tabular}
\end{center}
\caption{(a) boxplots of the daily wind speed data for each meteorological stations over the period 2000-2008; (b)  boxplots of the daily wind speed data rescaled by the average over the considered period.}\label{fig:boxplots}

\end{figure}

\begin{figure}[H]
	\begin{center}
		\begin{tabular}{cc}
			\includegraphics[width=0.4\linewidth, height=0.3\textheight]
			{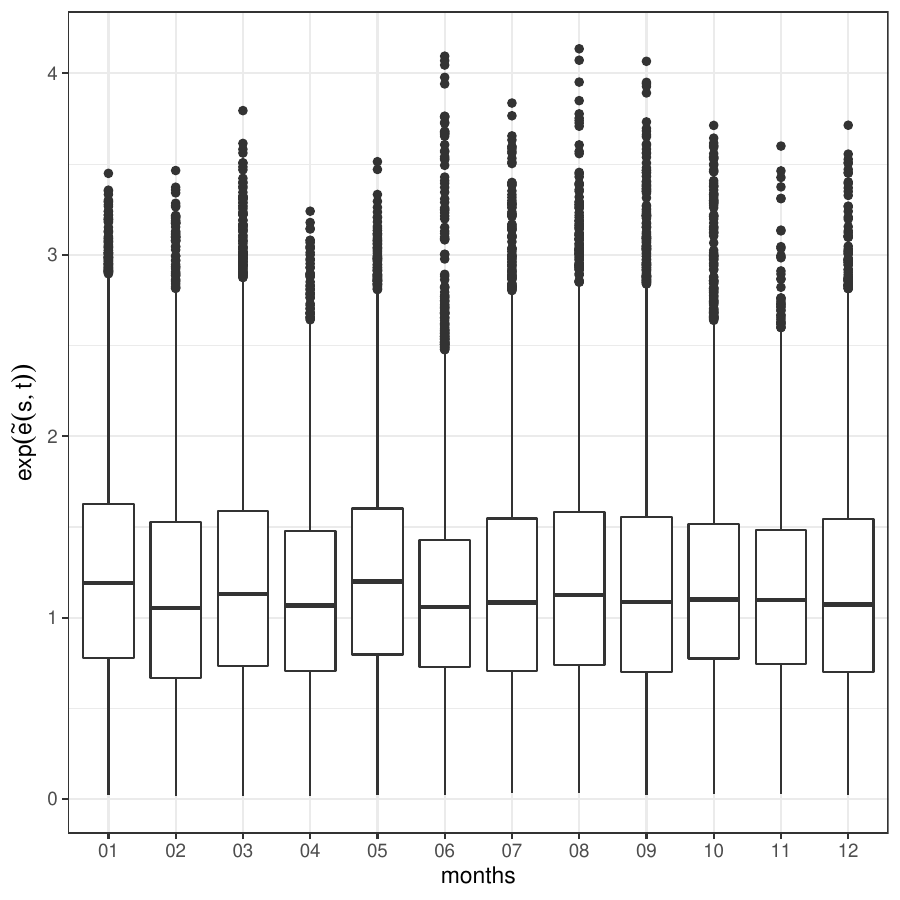}			 &
			\includegraphics[width=0.4\linewidth, height=0.3\textheight]
{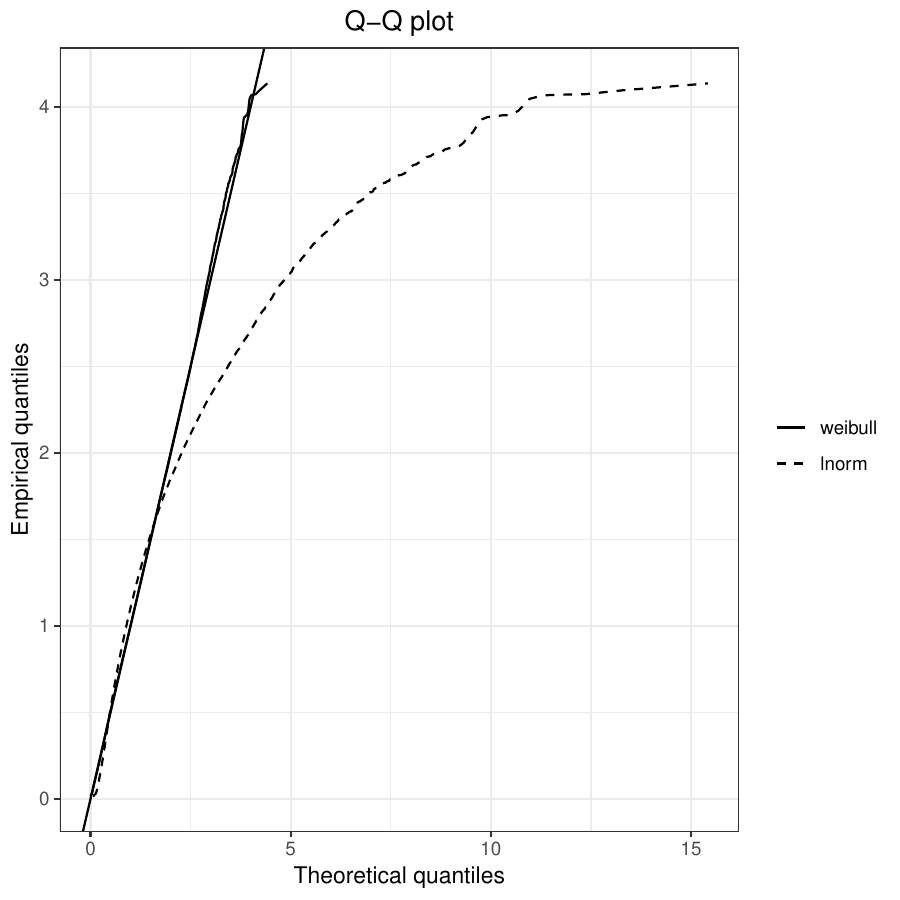}
			\\
			(a) & (b) \\
		\end{tabular} 
	\end{center}
	\caption{Preliminary analysis of residuals, $r(s,t)$  obtained in fitting model (\ref{eq:prel-reg}) by least-squares:
			(a) boxplots of the exponential of the residuals for each month.
		(b) qq-plot of the exponential of the residuals against the Weibull and Log-Gaussian distribution. 
	}
	\label{fig:res-1}
\end{figure}

\begin{figure}
\begin{center}
\begin{tabular}{cc}
\multicolumn{2}{c}{\includegraphics[width=0.4\linewidth, height=0.3\textheight]
{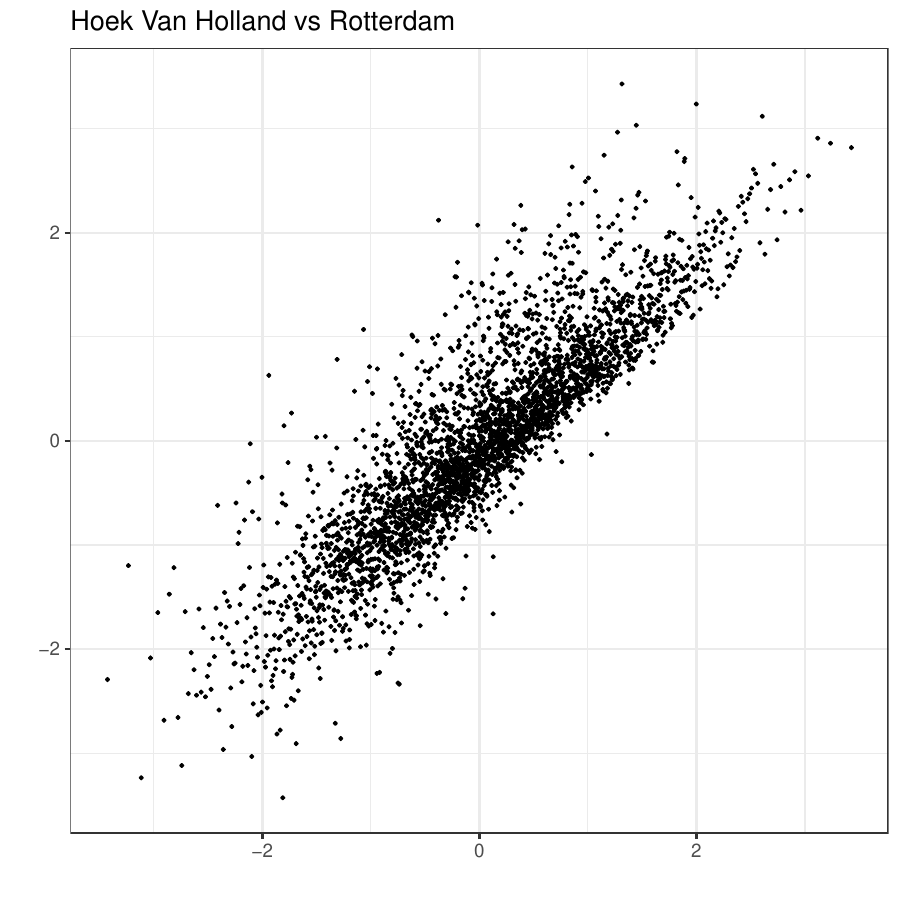}}
\\
\includegraphics[width=0.4\linewidth, height=0.3\textheight]
{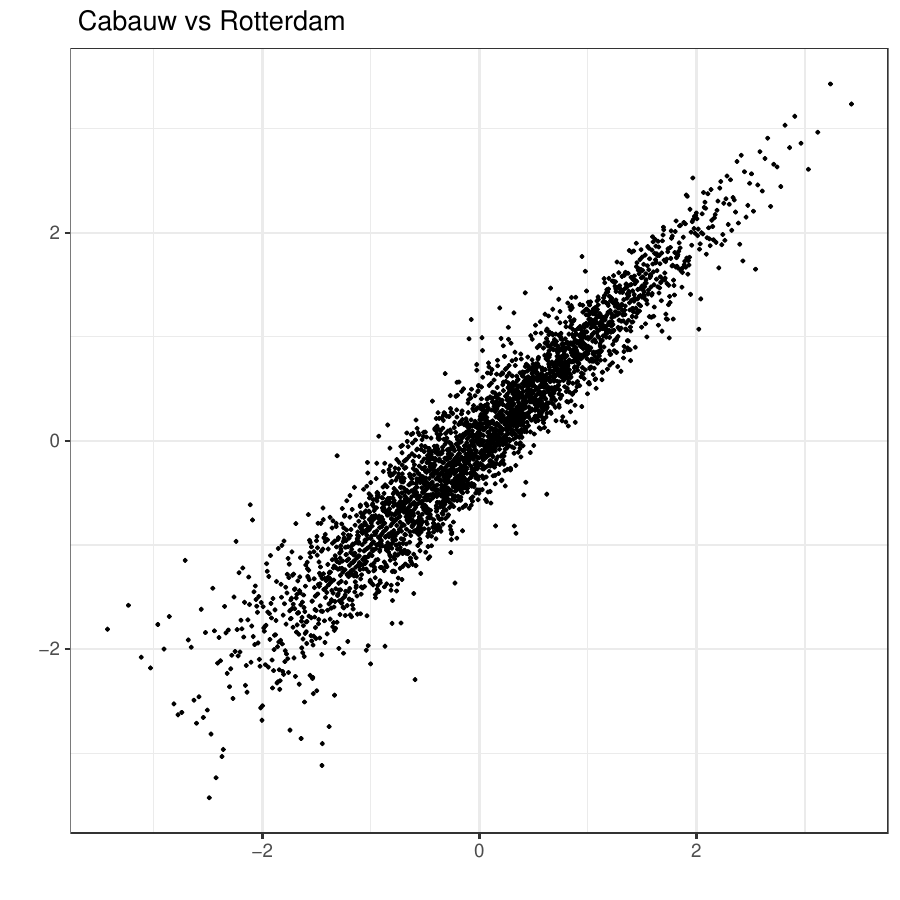}
&
\includegraphics[width=0.4\linewidth, height=0.3\textheight]
{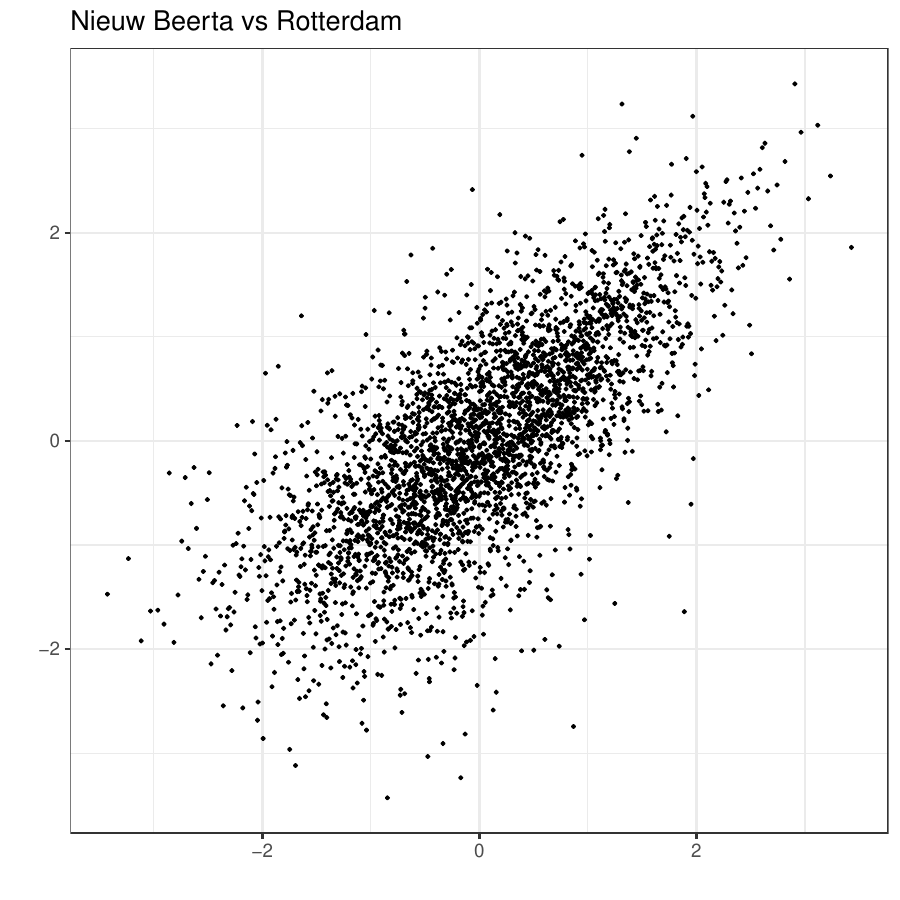}\\
\end{tabular}
\end{center}
\caption{Preliminary analysis of residuals, $r(s,t)$  obtained in fitting model (\ref{eq:prel-reg}) by least-squares:
the scatterplots of the normal scores of Rotterdam
station vs the  normal scores of three other stations.
}
\label{fig:res}
\end{figure}

\begin{figure}
\begin{center}
\hspace{-1cm}
\begin{tabular}{cc}
\hspace{-0.5cm}
\includegraphics[width=0.45\linewidth, height=0.3\textheight]{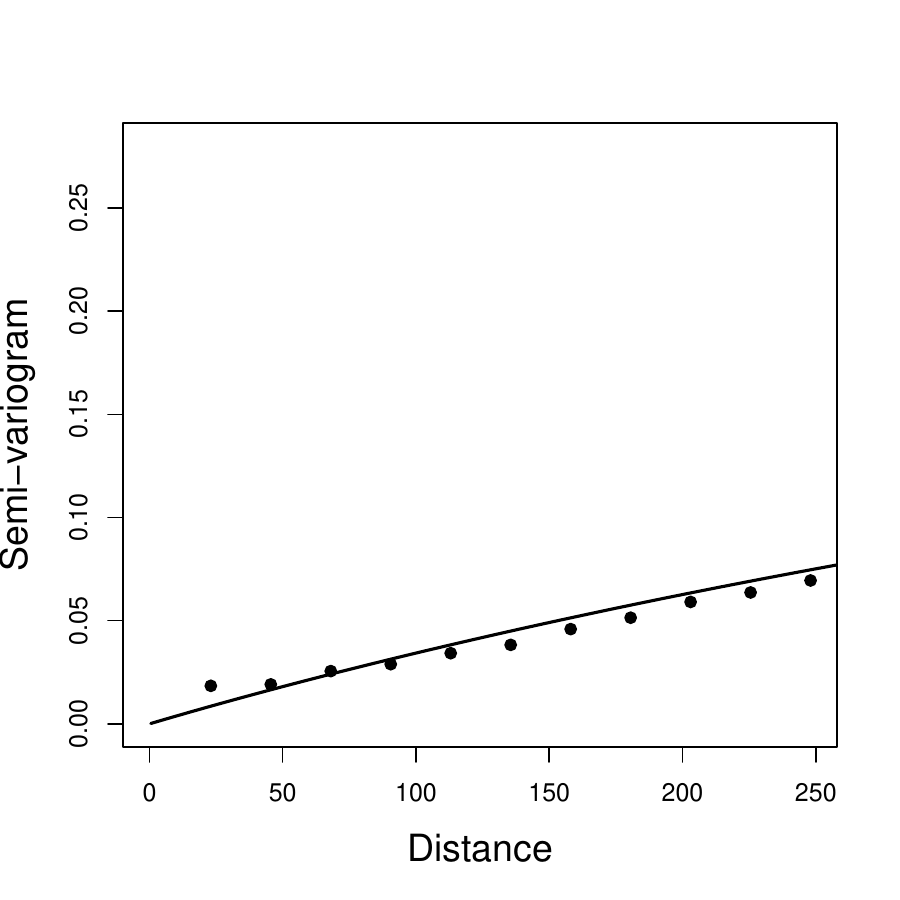}  &
\hspace{-0.5cm}
\includegraphics[width=0.45\linewidth, height=0.3\textheight]{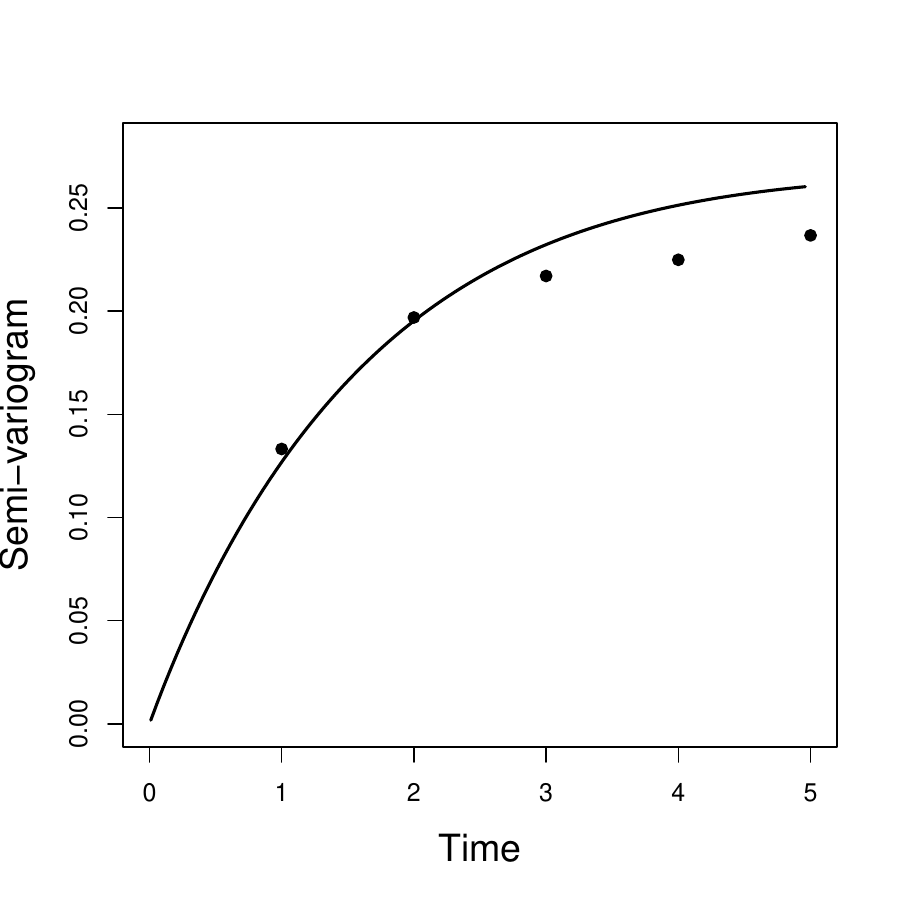}  \\
\end{tabular}
\end{center}
\caption{
 From left to right: Empirical spatial and temporal marginal semi-variograms  of the residuals (dotted points) and estimated theoretical 
 counterparts (solid line).
}\label{fig:res-variog}
\end{figure}

Using the preliminary estimates of the regression parameters
as starting values in the BFGS   \citep{FR:1987} optimization algorithm 
 we have fitted the Weibull and Log-Gaussian models  with WPL using seven years  (2000-2007).
 The last year has been used  for the evaluation of the  prediction performance using time-forward predictions.
Note that the sample size ($87630$ observations) prevents the use of a full likelihood approach even for the Log-Gaussian model. 

The estimation
 of regression and dependence  parameters for the six models is based on WPL,
with  a cut-off  weight $ c_{ij}$ equal to one  if %
$|t_i- t_j |\le 1$ and zero otherwise.
Table \ref{tab:estimation-results} collects the results of the estimation stage including the standard error estimates obtained 
with a sub-sampling
technique as in \citet{Bevilacqua:Gaetan:Mateu:Porcu:2012}. 
 {As one could}  expect, there is no big difference in trend estimates
among different models and correlation functions.	
However considering the PLIC criterion, our preference goes to the Weibull model with $\phi_{ST}=0$.

For the Weibull case with $\phi_{ST}=0$, using the MWPL estimates of the regression parameters
we first compute the  estimated residuals
and then  we compute the  empirical spatio-temporal semi-variogram of the residuals.
We compare it with the estimated theoretical semi-variogram obtained 
 {plugging  the MWPL estimates} into the theoretical  spatio-temporal  semi-variogram
i.e.  $\gamma_W(h,u)=\sigma_W^2(1-\rho_W(h,u))$
with $\rho_W$ given by:
 \begin{equation}\label{ccc2}
 \rho_{W}(h,u)=\frac{ \nu^{-2}(\kappa)    }{\left[\Gamma\left(1+{2}/{\kappa}\right)-\nu^{-2}(\kappa) \right]}\left[{}_2F_1\left(-{1/\kappa},-{1/\kappa};1;\rho^2(h,u)\right)-1\right].
 \end{equation}
 Figure \ref{fig:res-variog} shows the good agreement of  the estimated  theoretical  spatial and temporal marginal semi-variograms (i.e. the estimation of  $\gamma_W(h,0)$ and  $\gamma_W(0,u)$  respectively)  with the empirical 
 counterparts.

\begin{table}
\begin{center}
\scalebox{0.95}{
\begin{tabular}{|c|c|c|c|c|c|c|c|c|c|}
\hline \multicolumn{1}{|c|}{} & \multicolumn{2}{c|}{$\phi_{ST}=0$} &
\multicolumn{2}{c|}{$\phi_{ST}=0.5$} &\multicolumn{2}{c|}{$\phi_{ST}=1$}\\
\cline{1-7}
&  Weibull & Log-Gaussian & Weibull &  Log-Gaussian &  Weibull &  Log-Gaussian \\
\hline
${\beta_0}$&$-0.0222$&$0.0166$&$-0.0222$&$0.0168$&$-0.0221$&$0.0170$\\
&$(0.0026)$&$(0.0015)$&$(0.0026)$&$(0.0015)$&$(0.0026)$&$(0.0016)$\\
\hline
${\beta}_{1,1}$&$0.0747$&$0.0787$&$0.0747$&$0.0787$&$0.0747$&$0.0787$\\
&$(0.0025)$&$(0.0024)$&$(0.0025)$&$(0.0026)$&$(0.0025)$&$(0.0029)$\\
\hline
${\beta}_{2,1}$&$0.1822$&$0.1995$&$0.1822$&$0.1996$&$0.1822$&$0.1996$\\
&$(0.0030)$&$(0.0028)$&$(0.0029)$&$(0.0028)$&$(0.0029)$&$(0.0030)$\\
\hline

${\beta}_{1,2}$&$-0.0087$&$-0.0270$&$-0.0087$&$-0.0270$&$-0.0087$&$-0.0270$\\

&$(0.0567)$&$(0.0192)$&$(0.0566)$&$(0.0190)$&$(0.0566)$&$(0.0200)$\\

\hline

${\beta}_{2,2}$&$0.0138$&$0.0107$&$0.0138$&$0.0107$&$0.0137$&$0.0107$\\

&$(0.0306)$&$(0.0489)$&$(0.0306)$&$(0.0484)$&$(0.0306)$&$(0.0509)$\\

\hline

${\beta}_{1,3}$&$0.0274$&$0.0237$&$0.0274$&$0.0237$&$0.0274$&$0.0237$\\

&$(0.0192)$&$(0.0229)$&$(0.0192)$&$(0.0224)$&$(0.0192)$&$(0.0234)$\\

\hline

${\beta}_{2,3}$&$-0.0339$&$-0.0519$&$-0.0338$&$-0.0519$&$-0.0338$&$-0.0519$\\

&$(0.0101)$&$(0.0110)$&$(0.0100)$&$(0.0110)$&$(0.0100)$&$(0.0116)$\\

\hline

${\beta}_{1,4}$&$0.0093$&$0.0273$&$0.0093$&$0.0273$&$0.0093$&$0.0273$\\

&$(0.0548)$&$(0.0215)$&$(0.0548)$&$(0.0213)$&$(0.0548)$&$(0.0224)$\\

\hline

${\beta}_{2,4}$&$0.0042$&$0.0110$&$0.0042$&$0.0110$&$0.0042$&$0.0110$\\

&$(0.1238)$&$(0.0526)$&$(0.1238)$&$(0.0522)$&$(0.1238)$&$(0.0549)$\\

\hline
${\kappa}$&$2.0265$&$$&$2.0264$&$$&$2.0263$&\\

&$(0.0264)$&$$&$(0.0257)$&$$&$(0.0255)$&\\

\hline

\multirow{2}{*}{${\sigma}^2$}&$$&$0.3855$&$$&$0.3858$&$$&$0.3862$\\

&$$&$(0.0009)$&$$&$(0.0009)$&$$&$(0.0011)$\\

\hline

\multirow{2}{*}{${\phi}_S$}&$4067.21$&$1066.277$&$4071.738$&$1072.0239$&$4076.578$&$1078.6964$\\

&$(89.2924)$&$(3.4777)$&$(61.9349)$&$(3.4782)$&$(50.1251)$&$(3.3496)$\\

\hline

\multirow{2}{*}{${\phi}_T$}&$12.2794$&$4.9687$&$12.4249$&$5.1731$&$ 12.5715$&$5.3820$\\

&$(0.4035)$&$(0.0480)$&$(0.4057)$&$(0.0529)$&$(0.4080)$&$(0.0532)$\\
\hline\hline
PLIC&$8864239$&$10392021$&$8864463$&$10392832$&$8864821$&$10405428$\\
\hline
\end{tabular}
}
\end{center}
\caption{MWPL estimates for Weibull and Log-Gaussian models for the correlation model (\ref{eq:space-time-corr}). The standard error of the estimates are reported between the parentheses.} \label{tab:estimation-results}
\end{table}

We want to further evaluate the predictive performances of the proposed model 
by considering one-day ahead  predictions for the wind speed at the thirty meteorological stations but we have  limited the number of predictor variables  due to the computational load. Specifically, the predictor variables are the 150 wind speeds  observed during the past five days at the  stations. 

For the Weibull models,
we used the  simple kriging
predictor (\ref{pred-sk}). For the Log-Gaussian models, we have chosen the conditional expectation given the past observations \citep[formula 2]{DeOliveira:2006}.
In both cases, the predictions are obtained by plugging in the estimated parameters in the formulas.
As benchmark, we have also considered the na\"ive predictor $\widehat{Y}(s_i,t)=y(s_i,t-1)$,  that uses the   observation recorded the day before at the station.
The prediction performances are compared looking to the  root-mean-square prediction error (RMSE)  and the mean absolute prediction error (MAE).

In addition, we  considered the sample mean of the continuous ranked
probability score (CRPS)
 to evaluate the marginal  predictive distribution performance
\citep{Gneiting:Raftery:2007}. For a single predictive cumulative distribution function $F$ and a verifying observation $y$, the score  is defined as
$$\mathrm{CRPS}(F_,y)=\int\limits_{-\infty}^{\infty}(F(t)-\bm{1}_{[y,\infty]}(t))^2dt.
$$
For  a Weibull distribution, we have derived an analytical expression  of the corresponding score (see the Appendix) which turns out to be,
under our parametrization:
$$ \mathrm{CRPS}_W(F_{\kappa,\mu\nu(\kappa)}
,y)=y\left\{2\,(1-\exp\{-(y/\mu\nu(\kappa))^{\kappa}\}-1\right\}+
2\mu\left[2^{-1/k}-\nu(\kappa)\gamma\left(1+\frac{1}{\kappa},\frac{y^{\kappa}}{(\mu\nu(\kappa))^\kappa}\right)\right],
$$
where  $\gamma (s,z)=\int _{0}^{z}t^{s-1}\,\mathrm {e} ^{-t}\,{\rm {d}}t$ is the lower Gamma incomplete function.

\citet{Sandor:Leirch:2015} derived  the corresponding one for a Log-Gaussian random variable
$Y=\exp(\alpha+\beta Z)$ with cdf $F_{\alpha,\beta}$,
where $Z$ is  standard Gaussian random variable. Under our parametrization:
$$\mathrm{CRPS}_{LG}(F_{\mu-\sigma^2/2,\sigma},y)=
y\left[2\Phi(l(y)-1)\right]+2e^{\mu}\left[1-\Phi\left(\sigma/\sqrt{2}\right)-\Phi\left(l(y)-\sigma\right)\right],
$$
where $\Phi(\cdot)$ is the CDF of the standard Gaussian distribution and $l(y)=[\log(y)-(\mu-\sigma^2/2)]/\sigma$.

As a general consideration  the prediction based on a model (Weibull or Log-Gaussian) outclasses always the na\"ive prediction (see Table \ref{tab:prediction}).  Moreover, even though the simple kriging predictor is a suboptimal solution, the  Weibull model   outperforms  the Log-Gaussian model in terms of RMSE, MAE. Finally 
 also the CRPS of the Weibull model  outperform the Log-Gaussian model. Note that CRPS values do not  dependent on $\phi_{ST}$. This is not surprising since 
 the estimated marginal parameters for both models are very similar for $\phi_{ST}=0, 0.5, 1$.
 
Among the fitted covariance models, we give again a  preference to the correlation function  with $\phi_{ST}=0$.

\begin{table}[th!]
\begin{center}
\scalebox{1}{
\begin{tabular}{c|ccc|ccc|ccc|}
\cline{2-10}%
&\multicolumn{3}{c|}{$\phi_{ST}=0$} &\multicolumn{3}{c|}{$\phi_{ST}=0.5$}&\multicolumn{3}{c|}{$\phi_{ST}=1$}\\
\cline{2-10}
&RMSE    &   MAE  & CRPS &    RMSE &      MAE & CRPS&    RMSE &      MAE&CRPS\\
\cline{2-10}
W &  0.4461& 0.3486&0.3057 &0.4469 & 0.3491 &0.3057& 0.4502&0.3503&0.3057
\\
LG & 0.4517 & 0.3555&0.3068 & 0.4555 & 0.3585&0.3068 & 0.4611 & 0.3629&0.3068\\
\cline{2-10}%
\cline{2-3}
Na\"ive& \multicolumn{9}{c|}{MAE=  0.5137, RMSE= 0.4021}\\
\cline{2-10} 
\end{tabular}}
\caption{Preditiction performances for the Weibull and Log-Gaussian models for different space time interaction.}\label{tab:prediction}
\end{center}
\end{table}

\section{Concluding remarks}\label{sec:conclusions}

Motivated by a spatio-temporal analysis of  daily wind speed data 
from a network of meteorological stations in the Netherlands,
 {we proposed}  a non-stationary stochastic process  with Weibull marginal distributions  for  regression and dependence analysis when dealing with  positive continuous data.
In contrast to a Gaussian copula or, more  {generally}, to monotonic  transformations of a Gaussian process, 
our model offers a workable solution in the presence of different  dependence in the lower and upper distribution tails, $i.e.$ reflection asymmetry.

Additionally, we have shown that nice properties such as stationarity, mean-square
continuity and degrees of mean-square differentiability are
inherited  from the `parent' Gaussian random process. However,
 discontinuity of the paths can be easily induced by choosing a discontinuous correlation function for the 'parent' Gaussian process.

We also remark that even though we have  {limited  ourselves} to  continuous Euclidean space, our models can be extended
 to  a spherical domain  \citep{gneiting2013,PBG:2016} or to  a network space. In this respect the $X_m$ random process should represent a generalization of the  model in \citet{Warren:1992}.

A common drawback for the proposed model  is the lack of an amenable expression of the density outside of the bivariate case
that prevents an inference approach based on  likelihood methods
 and the derivation of an optimal predictor that minimizes  the mean square prediction error.
We have shown  with some numerical experiments that an inferential approach based on the
pairwise likelihood is an effective  solution for estimating the unknown parameter.
On the other hand  probabilities of multivariate events could be evaluated by Monte Carlo method since the random processes can be quickly simulated. However, our solution to the conditional prediction, based on a linear predictor,  is limited and deserves further consideration even if, in our simulations and real data example, it has been performed well.

\section*{Acknowledgements}
Partial support was provided by FONDECYT grant 1200068, Chile
and  by Millennium
Science Initiative of the Ministry
of Economy, Development, and
Tourism, grant "Millenium
Nucleus Center for the
Discovery of Structures in
Complex Data"
for Moreno Bevilacqua and by Proyecto de Iniciaci\'on Interno
DIUBB 173408 2/I de la Universidad del B\'io-B\'io for Christian Caama\~no.
\newpage
\bibliographystyle{biom} 
\bibliography{newbib}
\newpage
\appendix
\section*{Appendix}

In the sequel we will exploit   the identity for the 
hypergeometric function ${}_0F_1$,	
$${}_0F_1(;b;x)=\Gamma(b)x^{(1-b)/2}I_{b-1}(2\sqrt{x}).$$ 
where $I_{a}(x)$  is the modified Bessel function of the
first kind of order $a$.

\begin{proposition}\label{def:prop1}
	The $(a,b)-th$ product moment of
	 any pairs $W_1:=W(s_1)$ and $W_2:=W(s_2)$ is given by
	
	\begin{eqnarray}\label{momentweib}
	\E({W}_1^a\,{W}_2^b)&=&\frac{	\Gamma\left(1+{a}/{\kappa}\right)\Gamma\left(1+{b}/{\kappa}\right)}{\Gamma\left(1+{1/\kappa}\right)^{a+b}}
	{}_2F_1\left(-{a}/{\kappa},-{b}/{\kappa};1;
	\rho^2\right)
	\end{eqnarray}
	where $\rho=\rho(s_1-s_2)$
\end{proposition}

\begin{proof}
Using the series expansion of hypergeometric function ${}_0F_1$,	we have:
\begin{eqnarray}\label{cal12}
	\E({W}_1^a\,{W}_2^b)&=&
	\frac{\kappa^{2}\Gamma\left(1+{1/\kappa}\right)^{2\kappa}}{1-\rho^2}
	\int\limits_{0}^\infty\int\limits_{0}^\infty {u}^{\kappa+a-1}{v}^{\kappa+b-1}
	\exp\left\{-\frac{\Gamma\left(1+{1/\kappa}\right)^{\kappa}}{(1-\rho^2)}
	\left({u}^{\kappa}+
	{v}^{\kappa}
	\right)
	\right\}
	\nonumber\\
		&&
\quad \times{}\,_0F_1\left(1;\frac{\rho^2({u}{v})^{\kappa}\Gamma\left(1+{1/\kappa}\right)^{2\kappa}}{(1-\rho^2)^2}\right)d{u}d{v}\nonumber\\
	&=&\frac{\kappa^{2}\Gamma\left(1+{1/\kappa}\right)^{2\kappa}}{1-\rho^2}
	\sum\limits_{m=0}^{\infty}\quad \frac{1}{ { {(m!)^2}} }
	\left(\frac{\rho^2\Gamma\left(1+{1/\kappa}\right)^{2\kappa}}{(1-\rho^2)^2}\right)^m \nonumber\\
		&&\times
	\int\limits_{0}^\infty\int\limits_{0}^\infty
	\exp\left\{-\frac{\Gamma\left(1+{1/\kappa}\right)^{\kappa}}{(1-\rho^2)}
	\left({u}^{\kappa}+ v^{\kappa}
	\right)
	\right\}d{u}d{v}
	\nonumber\\
	&=&\frac{\kappa^{2}\Gamma\left(1+{1/\kappa}\right)^{2\kappa}}{1-\rho^2}
	\sum\limits_{m=0}^{\infty}\frac{I(m)}{ {(m!)^2}}\left(\frac{\rho^2\Gamma\left(1+{1/\kappa}\right)^{2\kappa}}{(1-\rho^2)^2}\right)^m
	\end{eqnarray}
Using Fubini's Theorem and (3.381.4) in
	\cite{Gradshteyn:Ryzhik:2007}, we obtain
	\begin{eqnarray}\label{res112}
	I(m)&=&\int\limits_{0}^{\infty}{u}^{\kappa+a+m\kappa-1}\
	\exp\left\{-\frac{\Gamma\left(1+{1/\kappa}\right)^{\kappa}}{(1-\rho^2)}u^{\kappa}\right\}
	d{u}
	\nonumber\\
	&&\times
	\int\limits_{0}^{\infty}{v}^{\kappa+b+m\kappa-1}
	\exp\left\{-\frac{\Gamma\left(1+{1/\kappa}\right)^{\kappa}}{(1-\rho^2)}v^{\kappa}\right\}
	d{v}
	\nonumber\\
	&=&\kappa^{-2}\Gamma\left(1+{a}/{\kappa}+m\right)\Gamma\left(1+{b}/{\kappa}+m\right)
	\nonumber\\
	&&
	\times
	\left(\frac{1-\rho^2}{\Gamma\left(1+{1/\kappa}\right)^{\kappa}}\right)^{1+{a}/{\kappa}+m}
	\left(\frac{1-\rho^2}
	{\Gamma\left(1+{1/\kappa}\right)^{\kappa}}\right)^{1+{b}/{\kappa}+m}
	\end{eqnarray}
	Combining equations  (\ref{cal12}) and (\ref{res112}), we
	obtain
	\begin{eqnarray*}
		\E({W}_1^a\,{W}_2^b)&=&
		\frac{(1-{\rho^2})^{1+(a+b)/\kappa}\Gamma\left(1+{a}/{\kappa}\right)\Gamma\left(1+{b}/{\kappa}\right)}{\Gamma\left(1+{1/\kappa}\right)^{a+b}}
		 \\
		&&\times\,
		{}_2F_1\left(1+{a}/{\kappa},1+{b}/{\kappa};1;{\rho^2}\right)\\
	\end{eqnarray*}
		Finally, using Euler transformation, we obtain (\ref{momentweib}).
\end{proof}

\begin{proposition}\label{def:prop2}
Let $s_1<s_2<\cdots< s_{n}<s_{n+1}$, with $s_i \in \R$.
	For the Weibull process $Y$ with underlying  exponential correlation function, the conditional expectation of $Y^a(s_{n+1})$,  $a>0$, given 
	$Y(s_{1})=y_1,\ldots,Y(s_{n})=y_n)$ is
\begin{equation*}
\begin{split}
\E(Y^{a}(s_{n+1})|Y(s_{1})=y_1,\ldots,Y(s_{n})=y_n)&=
\Gamma\left(\frac{a}{\kappa}+1\right)(1-\rho^2_{n,n+1})^{a/\kappa}[\nu(\kappa)\mu_{n+1}]^{a}\\
&\times \exp\left\{-\frac{y^{\kappa}_n}{(1-\rho^2_{n-1,n})[\nu(\kappa)\mu_n]^\kappa} 
\left[
\frac{   (1-\rho^2_{n-1,n} \rho^2_{n,n+1})}{(1-\rho^2_{n,n+1})}-1 \right]
\right\}\\
&\times
{}_1F_1\left(\frac{a}{\kappa}+1;1;\frac{\rho^2_{n,n+1}y_n^{\kappa}}{[\nu(\kappa)\mu_n]^{\kappa}(1-\rho^2_{n,n+1})}\right)
\end{split}
\end{equation*}
\end{proposition}
\begin{proof}

First,
note that using   (\ref{qqq}), the  density of the random variable 
$Y(s_{n+1})|(Y(s_{1})=y_1,\ldots,Y(s_{n})=y_n)$
is easily obtained as:

\begin{equation*}
\begin{split}
f(y_{n+1}|y_1,\ldots,y_n)&=\frac{\kappa y_{n+1}^{\kappa-1}}{\nu^{ \kappa}(\kappa)\mu_{n+1}^{\kappa}(1-\rho^2_{n,n+1})}
\exp\left\{-\frac{1}{(1-\rho^2_{n,n+1})}\left[\frac{y_{n+1}}{\nu(\kappa)\mu_{n+1}}\right]^{\kappa}\right\}\\
&\times \exp\left\{-\frac{y^{\kappa}_n}{(1-\rho^2_{n-1,n})[\nu(\kappa)\mu_n]^\kappa} 
\left[
\frac{(1-\rho^2_{n-1,n}\rho^2_{n,n+1})}{(1-\rho^2_{n,n+1})}-1\right]\right\}\\
&\times I_{0}\left(\frac{2|\rho_{n,n+1}|(y_ny_{n+1})^{\kappa/2}}{\nu^{\kappa}(\kappa)(\mu_n\mu_{n+1})^{\kappa/2}(1-\rho^2_{n,n+1})}\right).
\end{split}
\end{equation*}

Using the series expansion of
	hypergeometric function ${}_0F_1$, we obtain:
\begin{equation}\label{expa}
\begin{split}
\E(Y^{a}(s_{n+1})&|Y(s_{1})=y_1,\ldots,Y(s_{n})=y_n)=\frac{\kappa}{\nu^{ \kappa}(\kappa)\mu_{n+1}^{\kappa}(1-\rho^2_{n,n+1})}\\
&\times \exp\left\{-\nu^{-\kappa}(\kappa)\left[
\frac{(1-\rho^2_{n-1,n}\rho^2_{n,n+1})y_n^{\kappa}}{\mu_{n}^{\kappa}(1-\rho^2_{n-1,n})(1-\rho^2_{n,n+1})}-\frac{y^{\kappa}_n}{\mu_{n}^{\kappa}(1-\rho^2_{n-1,n})}\right]\right\}\\
&\times\int\limits^{\infty}_{0}y_{n+1}^{\kappa+a-1}e^{-\frac{1}{(1-\rho^2_{n,n+1})}\left[\frac{y_{n+1}}{\nu(\kappa)\mu_{n+1}}\right]^{\kappa}}
{}_0F_1\left(;1;\frac{\rho^2_{n,n+1}(y_ny_{n+1})^{\kappa}}{\nu^{2\kappa}(\kappa)(\mu_n\mu_{n+1})^{\kappa}(1-\rho^2_{n,n+1})^2}\right)dy_{n+1}\\
&=\frac{\kappa}{\nu^{ \kappa}(\kappa)\mu_{n+1}^{\kappa}(1-\rho^2_{n,n+1})}
\\
&\times \exp\left\{-\nu^{-\kappa}(\kappa)\left[
\frac{(1-\rho^2_{n-1,n}\rho^2_{n,n+1})y_n^{\kappa}}{\mu_{n}^{\kappa}(1-\rho^2_{n-1,n})(1-\rho^2_{n,n+1})}-\frac{y^{\kappa}_n}{\mu_{n}^{\kappa}(1-\rho^2_{n-1,n})}\right]\right\}\\
&\times \sum\limits_{m=0}^{\infty}\frac{I(m)}{{ {(m!)^2}}}\left(\frac{\rho^2_{n,n+1}y_n^{\kappa}}{\nu^{2\kappa}(\kappa)(\mu_n\mu_{n+1})^{\kappa}(1-\rho^2_{n,n+1})^2}\right)^m
\end{split}
\end{equation}
where
\begin{eqnarray}\label{expb}
I(m)&=&\int\limits^{\infty}_{0}
y_{n+1}^{\kappa+a+\kappa m-1}
\exp\left\{-\frac{1}{(1-\rho^2_{n,n+1})}\left[\frac{y_{n+1}}{\nu(\kappa)\mu_{n+1}}\right]^{\kappa}\right\}dy_{n+1}\nonumber\\
&=&\kappa^{-1}[(1-\rho^2_{n,n+1})\nu^{\kappa}(\kappa)\mu^{\kappa}_{n+1}]^{a/\kappa+m+1}\Gamma\left(\frac{a}{\kappa}+m+1\right)
\end{eqnarray}
Combining equations (\ref{expa}) and (\ref{expb}), we obtain the conditional expectation in  proposition \ref{def:prop2}.
\end{proof}

\begin{proposition}\label{def:prop3}
	The $\operatorname{CRPS}$ associated with the  $Weibull(\alpha,\beta)$ distribution %
	is given by
	\begin{eqnarray}\label{crpsweib} \operatorname{CRPS}(F_{\alpha,\beta},y)=y[2F_{\alpha,\beta}(y)-1]-2\beta\gamma\left(1+\frac{1}{\alpha},\frac{y^{\alpha}}{\beta^{\alpha}}\right)+2^{-1/\alpha}\beta\Gamma\left(1+\frac{1}{\alpha}\right)
	\end{eqnarray}
	where $F_{\alpha,\beta}(y)=1-\exp^{-(y/ \beta)^{\alpha}}$ and 
	${\displaystyle \gamma (s,z)=\int _{0}^{z}t^{s-1}\,\mathrm {e} ^{-t}\,{\rm {d}}t,\,s>0}$
	is the lower incomplete gamma function.
\end{proposition}
\begin{proof}
We first note that the $\operatorname{CRPS}$ can also be written as
\begin{equation*}
\operatorname{CRPS}(F,y)=\E_{F}|Y-y|-\frac{1}{2}\E_{F}|Y-Y'|
\end{equation*}
where $Y$ and $Y'$ are independent random
variables with cumulative distribution function $F$ and finite first moment.
The first term can be integrated out using
the properties of the Weibull density, yielding
\begin{eqnarray*}
\E_{F}|Y-y|&=& \int\limits_{-\infty}^{y}(y-t)f_{\alpha,\beta}(t)dt-\int\limits^{\infty}_{y}(y-t)f_{\alpha,\beta}(t)dt\nonumber\\
&=& yF_{\alpha,\beta}(y)-\int\limits_{-\infty}^{y}\frac{\alpha}{\beta^{\alpha}}t^{\alpha}\exp\left[-\left(\frac{t}{\beta}\right)^{\alpha}\right]dt
  -y[1-F_{\alpha,\beta}(y)]\nonumber\\
  &+&\int\limits_{y}^{\infty}\frac{\alpha}{\beta^{\alpha}}t^{\alpha}\exp\left[-\left(\frac{t}{\beta}\right)^{\alpha}\right]dt\nonumber\\
&=& y[2F_{\alpha,\beta}(y)-1]-\beta\gamma\left(1+\frac{1}{\alpha},\frac{y^{\alpha}}{\beta^\alpha}\right)+\beta\Gamma\left(1+\frac{1}{\alpha},\frac{y^{\alpha}}{\beta^{\alpha}}\right)
\end{eqnarray*}
where 
${\displaystyle \gamma (s,z)=\int _{0}^{z}t^{s-1}\,\mathrm {e} ^{-t}\,{\rm {d}}t,\,s>0}$
is the lower incomplete gamma function and \\
\noindent ${\displaystyle \Gamma(s,z) = \Gamma(s) - \gamma(s, z)}$
is the upper incomplete gamma function. We have:
\begin{equation*}
\E_{F}|X-y|=y[2F_{\alpha,\beta}(y)-1]-\beta\left[2\gamma\left(1+\frac{1}{\alpha},\frac{y^{\alpha}}{\beta^{\alpha}}\right)-\Gamma\left(1+\frac{1}{\alpha}\right)\right]
\end{equation*}
The second term can be calculated  using its relation to the Gini concentration ratio $G$:
\begin{equation*}
\E_{F}|Y-Y'|=\int_{\R_{+}^2}|y-y'|f_{\alpha,\beta}(y)f_{\alpha,\beta}(y')dy\,dy'=2\E(Y)G=2\beta\Gamma\left(1+\frac{1}{\alpha}\right)(1-2^{-1/\alpha})
\end{equation*}
Putting both terms together, we obtain
\begin{equation*}
\operatorname{CRPS}(F_{\alpha,\beta},y)=y[2F_{\alpha,\beta}(y)-1]-2\beta\gamma\left(1+\frac{1}{\alpha},\frac{y^{\alpha}}{\beta^{\alpha}}\right)+2^{-1/k}\beta\Gamma\left(1+\frac{1}{\alpha}\right).
\end{equation*}
\end{proof}

\end{document}